\documentclass[journal,onecolumn,twoside]{IEEEtran}   

\usepackage[latin1]{inputenc}
\usepackage{epsfig}
\usepackage{amssymb}
\usepackage{amsmath}
\usepackage{amsfonts}
\usepackage{amsbsy}
\usepackage{multirow,color}
\usepackage{psfrag,pstricks}

\newcommand{\bm}[1]{\mbox{$\bf #1$}}
\newcommand{\joint}{\mbox{\scriptsize joint}}
\newcommand{\simple}{\mbox{\scriptsize simple}}
\def\boldtheta{\boldsymbol{\theta}}
\def\boldcaptheta{\boldsymbol{\Theta}}
\def\pfp{\epsilon_{1}}

\newtheorem{lemma}{Lemma}
\newtheorem{proposition}{Proposition}
\newtheorem{conjecture}{Conjecture}

\newcommand{\seq}[2]{\mathbf{#1}_{#2}}
\newcommand{\symb}[3]{#1(#2,#3)}
\newcommand{\set}[1]{\mathcal{#1}}
\newcommand{\nchoosek}[2]{\left(_{#2}^{#1}\right)}
\newcommand{\Prob}[2]{\mbox{Pr}_{#1}[#2]}
\newcommand{\Probb}[1]{\mbox{Pr}_{#1}}
\newcommand{\expect}[2]{\mathbb{E}_{#1}\left[#2\right]}
\newcommand{\partialder}[1]{\frac{\partial}{\partial #1}}
\newcommand{\class}[2]{\set{P}^{#1}(#2)}

\title{Worst case attacks against binary probabilistic traitor tracing codes}
\author{Teddy Furon and Luis P{\'e}rez-Freire}

\begin{document}

\maketitle

\begin{abstract}
An insightful view into the design of traitor tracing codes should
necessarily consider the worst case attacks that the colluders can
lead. This paper takes an information-theoretic point of view
where the worst case attack is defined as the collusion strategy
minimizing the achievable rate of the traitor tracing code. Two
different decoders are envisaged, the joint decoder and the simple
decoder, as recently defined by P. Moulin~\cite{Moulin08universal}.
Several classes of colluders are defined with
increasing power. The worst case attack is derived for each class
and each decoder when applied to Tardos' codes and a probabilistic
version of the Boneh-Shaw construction. This contextual study
gives the real rates achievable by the binary probabilistic
traitor tracing codes. Attacks usually considered in literature, such as majority or minority votes, are indeed largely suboptimal.
This article also shows the utmost importance of the time-sharing concept in probabilistic codes.
\end{abstract}

\section{Introduction}  \label{sec:intro}
This article deals with traitor tracing which is also known as
active fingerprinting, content serialization, user forensics or
transactional watermarking. A typical application is, for instance, as
follows: A video on demand server distributes personal copies
of the same content to $n$ buyers. Some are dishonest users whose
goal is to illegally redistribute a pirate copy. The rights holder
is interested in identifying these dishonest users. For this
purpose, a unique user identifier consisting on a sequence of $m$
symbols is embedded in each video content thanks to a watermarking
technique, thus producing $n$ different (although perceptually
similar) copies. This allows tracing back which user has illegally
redistributed his copy. However, there might be a collusion of $c$
dishonest users, $c>1$. This collusion mixes their copies in order
to forge a pirated content which contains none of the identifiers
but a mixture of them.

The traitor tracing code invented by Gabor Tardos in 2003
\cite{Tardos2003:Optimal} becomes more and more popular. This code
is a probabilistic weak traitor tracing code, where the
probability of accusing an innocent is not null.
Its performance is usually evaluated in terms of the probability
$P_{FA}$ of accusing an innocent and the probability of missing all
colluders $P_{FN}$. Most of the articles dealing with the Tardos code
aim at finding a tighter lower bound on the length of the code. In
his seminal work, G. Tardos shows that, in order to guarantee
$P_{FA}<\pfp$ and $P_{FN}<\pfp^{c/4}$ as defined in the Boneh \& Shaw problem \cite{Boneh98:fingerprinting}, the code length must satisfy $m> 100 c^2\log
n/\pfp$. Many researchers found the constant $100$ too arbitrary.
Better approximations\footnote{Numbers are given for the non symmetric decoding, where symbols `0' are discarded as in the original Tardos setup.}
are $4\pi^2$\cite{Skoric2006:Tardos},
$85$ \cite{BlayerT08}, and $16$ \cite{CerouF08}.
A main improvement came from the symmetric decoding~\cite{Skoric2008:Symmetric}.
Other works propose more practical implementations of the Tardos
code~\cite{Nuida07IH}. The reader will also find a pedagogical presentation
of this code in~\cite{Furon08IH}.

Our article is very different than these past threads of studies
as we give the theoretical performances of the code whatever the
accusation algorithm.
In a nutshell, our work consists in applying
the results of~\cite{Moulin08universal}.
In this article, P. Moulin gives the definition of capacity for the traitor tracing problem, providing exact capacity expressions for the blind model, i.e. when the decoder does not know in advance neither the number of colluders nor the particular collusion strategy followed by them. So far, only bounds to the capacity had been derived by other authors (see references in~\cite{Moulin08universal}). In words, capacity is defined as the maximum (over all traitor tracing codes) of the minimum (over all strategies allowed by the collusion model) of an appropriate mutual information functional. Nevertheless, the problems of finding the best traitor tracing codes and the optimal collusion attacks are left open, although some important hints are given in~\cite{Moulin08universal} and more recently in~\cite{Amiri2009:High}.
Our results are not in the direction of solving this game-theoretic problem. We consider specific binary fingerprinting codes and seek for the collusion strategy minimizing the mutual information. 
Therefore, we cannot speak of capacity of a given collusion channel as in~\cite{Moulin08universal},
but of the maximum achievable rate of a given binary code. Our results are mainly aimed at providing more insight into the binary Tardos code, but the methodology can be easily extended, in general, to other code constructions based on the same principles. In fact, as explained in the sequel, our study also deals with a probabilistic version of the Boneh \& Shaw code \cite{Boneh98:fingerprinting}.

The goal of our study is twofold. First, it seems that an invariance
property governs the design of Tardos code~\cite{Furon08IH}: the Markov
lower bounds on code length derived
in~\cite{Tardos2003:Optimal,Skoric2006:Tardos,BlayerT08} involve means and
variances of the innocents and colluders scores which are invariant with
respect to the collusion strategy.
Therefore, this nuisance parameter unknown at the accusation side is no
longer a problem since the bounds hold whatever its value. A priori, there
is no
collusion attack which is worse than any others. This is yet only
true as far as the first and second order statistics of the scores
are concerned. Higher order statistics do not share this
invariance property. Furthermore, this invariance property only holds for
the decoder
originally suggested by G. Tardos.
On contrary, the achievable rate of the traitor tracing code is an
appropriate measure to quantify how damageable is a given collusion
process whatever the decoding algorithm.
Therefore, we are looking for the
worst case collusion attack for a given number of colluders minimizing
this quantity.
The code is deemed sound whenever its rate is below this minimal mutual
information, hence the term maximum achievable rate.
These results clearly show that classical assessment against, for instance,
majority or minority attacks can largely overestimate the
performance of the code because these are far from being the worst
collusion processes.

The second goal of this article is to show the importance of time-sharing, which has been already highlighted in the theoretical derivations of P. Moulin~\cite{Moulin08universal}, in practice when a binary Tardos code is considered. Time-sharing is a concept well known in multiuser information theory \cite{CT91}, by which using two or more codes of different rates a new code can be constructed by using each code in disjoint fractions of the time.
In the Tardos code, the probability of having a symbol `1' in a code sequence changes from one index $i$ to another according to a given auxiliary random variable $P$, which is indeed the ``time-sharing'' variable that selects the code to be used in each index. 
Therefore, the achievable rate of the codes studied in this paper is defined as an expectation of a function over the time-sharing random variable $P$. It is very interesting to plot this later function with respect to $P$. Some attacks succeed in canceling this function over a range. Therefore, the support and the values taken by the probability density function $f(P)$ of the
time-sharing variable is of utmost importance.
An appropriate time-sharing leads to huge improvements,
provided the time-sharing sequence remains secret for the colluders.
Moreover, this study also shows that even when this sequence is disclosed,
performing traitor tracing is still possible in theory as the rate never exactly cancels. An interesting byproduct of our analysis is that it indeed addresses the analysis of binary traitor tracing codes without time-sharing, which has not been addressed before from the information-theoretic viewpoint, specially in the case of the simple decoder.

We recently discovered that E. Amiri and G. Tardos~\cite{Amiri2009:High}, and independently Huang and Moulin~\cite{Huang09}, addressed the same issues. However, relatively few of their results cover exactly our propositions:
\begin{itemize}
\item For the joint decoder, they succeeded to derive in a game-theoretic setting the capacity-achieving parameter $f(P)$. This is indeed a probability mass function (pmf - i.e. the time-sharing variable is discrete) strongly dependent on the collusion size. However, in a real case scenario one cannot foresee the exact number of colluders: at most, a maximum collusion size can be anticipated. The code is guaranteed to perform well whenever the number of colluders is below the predicted maximum; however, for bigger collusion sizes the code becomes unreliable. Furthermore, the numerical computation of the optimal $f(P)$ seems not feasible for a large number of colluders \cite{Huang09}. This motivates the interest in the study of the probabilistic traitor tracing code with a fixed continuous $f(P)$ which, albeit suboptimal, does reasonably well for any collusion size. Remarkably, the $f(P)$ proposed by Tardos which we study in this paper seems to be very close to the optimal $f(P)$ when the number of colluders is very large, according to \cite{Huang09}, and the asymptotic loss with respect to the capacity is only within a small factor. 
\item For the simple decoder, Amiri and Tardos~\cite{Amiri2009:High} considered a scenario where all colluder identities were disclosed except one, and the decoder is looking for the identity of this unknown colluder. Our simple decoder is the one defined by P. Moulin \cite{Moulin08universal} which is very different and more realistic: no colluder is caught, and the goal of the decoder is to make a first accusation. 
\end{itemize}

Sec.~\ref{Sec:Model} introduces all the mathematical definitions
and assumptions needed to derive the worst case attacks: the type
of traitor tracing code we are dealing with (Sec.~\ref{sub:BinaryCode}), the introduction of
four different classes of colluders referred to as A, B, C and D (Sec.~\ref{sub:Collusion}),
and two possible accusation strategies based on the so-called
joint and simple decoders~(Sec.~\ref{sub:DecodingFamilies}).
This paper gives the worst case attacks that a given class of colluders can lead against a given family of decoders: joint decoder in Sec.~\ref{sec:Joint} and simple decoder in Sec.~\ref{sec:simple}.

\section{Model}	\label{Sec:Model}

\subsection{Notation}
\label{subsec:notation}

First of all we summarize the most important notational conventions to be used throughout the paper. Random variables and their realizations are denoted by capital and lowercase letters, respectively. Boldface letters denote column vectors. Calligraphic letters are reserved for sets. $\Prob{X}{x}$ is the probability that the discrete random variable $X$ takes the value $x$. The shorthand $[m]$ will be used to denote a sequence of indices $\{1,\ldots,m\}$. $H(.)$ denotes entropy of a discrete random variable. $h_b(x) = -x\log(x) -(1-x)\log(1-x)$ is the binary entropy. $D_{KL}(\Probb{X}||\Probb{Y})$ is the Kullback-Leibler divergence or relative entropy between the random variables $X$ and $Y$. $\log$, the logarithm to the base $2$, is preferably used in order to give all rates and entropies in bits, whereas $\ln$ is the natural logarithm.

\subsection{Binary probabilistic code with time-sharing}	\label{sub:BinaryCode}

We briefly remind how the Tardos code is designed, as an example of a probabilistic code with time-sharing.
The binary code $\set{X}$ is composed of $n$ sequences of $m$ bits. The sequence
$\seq{X}{j}=(\symb{X}{j}{1},\cdots,\symb{X}{j}{m})^T$ identifying
user $j$ is composed of $m$ independent binary symbols, with
$\Prob{\symb{X}{j}{i}}{1}=p_{i}$, $\forall i\in[m]$. 
The auxiliary random variables $\{P_{i}\}_{i=1}^{m}$ are independent
and identically distributed in the
range $[0,1]$ according to the probability density function $f(p)$: $P_{i}\sim f(p)$. Tardos pdf
$f(p)=\left(\pi\sqrt{p(1-p)}\right)^{-1}$ is symmetric around $1/2$:
$f(p)=f(1-p)$. It means that symbols `1' and `0' play a similar
role with probability $p$ or $1-p$. Both the code $\set{X}$ and the time-sharing sequence $\seq{p}{}=(p_1,\ldots,p_m)^T$ must remain as secret parameters. In the original paper, the pdf is slightly different as it is defined in $[t,1-t]$ where $t>0$ is the cut-off parameter fixed to $1/300c$. We do not consider this cut-off since the integrals are all well defined over $(0,1)$. 

This definition might encompass more fingerprinting codes than the Tardos one.
Although its construction is very different, the Boneh \& Shaw code (BS) shares a similar statistical structure~\cite{Boneh98:fingerprinting}.
When $n$ users are addressed, the ratio $P$ of symbol `1' in the code symbols  $\{\symb{X}{j}{i}\}_{j=1}^{n}$ for a given index $i\in[m]$ can be considered as a discrete random variable whose probability mass function is given by $\Prob{P}{k/n}=1/n,\,\forall k\in[n]$. This means that the sequence identifying user $j$ is composed of $m$ binary symbols, with
$\Prob{\symb{X}{j}{i}}{1}=p_i$, $\forall i\in[m]$, where $p_i \in \{1/n, 2/n,\ldots,(n-1)/n,1\}$ is chosen equiprobably. Therefore, the resemblance with the Tardos construction is clear: as $n$ goes to infinity, this code can be constructed as a Tardos code but with a flat pdf over $[0,1]$: $f(p)=1\;\forall\;p\in[0,1]$.

However, the difference between the Tardos and the BS codes is that the rate of the latter is imposed by construction.
Let us define the rate $R$ of a fingerprinting code by $R=\log(n)/m$. In a BS code, the rate is known to be $\log(n)/r(n-1)$, where $r$ is the so-called ``replication factor''~\cite{Boneh98:fingerprinting}. However, in order to perform a reliable accusation, the rate of any code must be lower than the capacity of the collusion channel~\cite{Moulin08universal}. Finding the capacity induced by a collusion process is a hard problem in general. This paper only deals with the achievable rate of Tardos-like codes (either with the Tardos pdf or a flat pdf to simulate a BS code), which is defined as the maximum rate guaranteeing a reliable decoding for any collusion process in a given class.

\subsection{Collusion process}	\label{sub:Collusion}
Denote the subset of colluder indices by
$\set{C}=\{j_{1},\cdots,j_{c}\}$, and
$\set{X}_\set{C}=\{\seq{X}{j_{1}},\ldots,\seq{X}{j_{c}}\}$
the restriction of the code to this subset. The collusion attack is the
process of taking sequences in $\set{X}_\set{C}$ as inputs and
yielding the pirated sequence $\seq{Y}{}$ as an output.

Traitor tracing codes have been first studied by the cryptographic
community and a key-concept is the marking assumption introduced
by Boneh and Shaw~\cite{Boneh98:fingerprinting}. It states that,
in its narrow-sense version, whatever the strategy of the
collusion $\set{C}$, we have
$Y(i)\in\{\symb{X}{j_{1}}{i},\cdots,\symb{X}{j_{c}}{i}\}$. In
words, colluders forge the pirated copy by assembling chunks from
their personal copies. It implies that if, at index $i$, the
colluders' symbols are identical, then this symbol value is
decoded at the $i$-th chunk of the pirated copy.

This is what watermarkers have understood from the pioneering cryptographic work. However, this has led to misconceptions. Another important thing is the way cryptographers have modeled a host content: it is a binary string where some symbols can be changed without spoiling the regular use of the content. These locations are used to insert the code sequence symbols. Cryptographers assume that colluders disclose codeword symbols from their identifying sequences comparing their personal copies symbol by symbol. The colluders cannot spot a hidden symbol if it is identical on all copies, hence the marking assumption.

In a multimedia application, for instance, the content is typically divided into chunks. A chunk can be a few second clip of audio or video. Symbol $\symb{X}{j}{i}$ is hidden in the $i$-th chunk of the content with a watermarking technique. This gives the $i$-th chunk sent to the $j$-th user. In this paper, we only address collusion processing where the pirated copy is forged by picking chunks from the colluders' personal copies. We do not cope with the mixing of several chunks into one (we assume that the watermarking technique is robust enough to handle this mixing collusion process). The marking assumption still holds but for another reason: as the colluders ignore the watermarking secret key, they cannot create chunks of content watermarked with a symbol they do not have.
However, contrary to the original cryptographic model, this also implies that the colluders might not know which symbol is embedded in a chunk.

\subsubsection{Mathematical model} \label{subsub:MathModelColl}
Our mathematical model of the collusion is essentially based on four main assumptions.
The first assumption is the memoryless nature of the collusion attack. Since the symbols of the code are independent, it seems relevant that the pirated sequence $\seq{Y}{}$ also shares this property. Therefore, the value of $Y(i)$ only depends on $\{\symb{X}{j_{1}}{i},\cdots,\symb{X}{j_{c}}{i}\}$.
The second assumption is the stationarity of the collusion process. Except when the Tardos code is broken (this is explained in the next section), we assume that the collusion strategy is independent of the index $i$ in the sequence. Therefore, we can describe it forgetting the index $i$ in the sequel.
The third assumption is that the colluders select the value of the symbol $Y$ depending on the values of their symbols, but not on their order. That is, the collusion channel is invariant to permutations of $\{\symb{X}{j_{1}}{i},\cdots,\symb{X}{j_{c}}{i}\}$. Therefore, the input of the collusion process is indeed the type of their symbols. In the binary case, this type is fully defined by the following sufficient statistic: the number $\Sigma_i$ of symbols `1': $\Sigma_i=\sum_{j=1}^c \symb{X}{j}{i}$. These three first assumptions greatly simplify the analysis of the problem without restricting the power of the colluders because they do not prevent them from  implementing an optimal collusion attack (see sections 2 and 3 of \cite{Moulin08universal}). Hence, our approach does not imply any loss of generality.  The fourth assumption is that the collusion process may not be deterministic, but random. These four assumptions yield that the collusion attack is fully described by the following parameter vector: $\boldtheta=(\theta_0,\ldots,\theta_c)^T$, with $\theta_\sigma = \Prob{Y}{1|\Sigma=\sigma}$. The following subsection gives examples of such collusion attacks, but we can already state that they all share the following property: The marking assumption enforces that $\theta_0=0$ and $\theta_c=1$. The authors of~\cite{Amiri2009:High} also speak about `eligible channel'.

\subsubsection{Classes of colluders}	\label{subsub:CollClass}
We introduce four classes of attacks with increasing power.

\paragraph{Class-A}
The weakest kind of colluders decides the value of the symbol
$Y(i)$ without considering all their symbols. Before receiving the
personal copies, these $c$ dishonest users have already agreed on
how to forge the pirated copy. This strategy amounts to set an
assignation sequence $(M_{1},\cdots,M_{m})$ with
$M_{i}\in\set{C}$, such that $Y(i)=\symb{X}{M_i}{i}$. We assume
that the colluders share the risk, so that the cardinality
$|\{i|M_{i}=j\}|\approx m/c$, for all $j\in\set{C}$. The
assignation sequence is random and independent of the personal
copies. Hence, for each collusion size, Class-A has a single collusion attack $\boldtheta$ given by $\theta_\sigma=\sigma/c$, $\forall\,\sigma=0,\ldots,c$.
For the sake of coherence with the subsequent notation, we say that $\boldtheta \in \class{A}{c}\triangleq\{c^{-1}(0,1,\ldots,c-1,c)^T\}$.

\paragraph{Class-B}
This second class of colluders differs from Class-A in the fact
that the assignation sequence is now a function of the personal
copies. These colluders are able to split their copies in chunks
and to compare them sample by sample. Hence, for any index $i$,
they are able to notice that, for instance, chunks $c_{j_{1}i}$
and $c_{j_{2}i}$ are different or identical. For binary embedded
symbols, they can constitute two stacks, each containing identical
chunks. This allows new collusion processes such as majority vote, minority vote, coin flip~\cite{Furon08IH}.

The important thing is that colluders can notice differences
between chunks, but they cannot tell which chunk contains symbol
`0'.\footnote{Note that in order this to be strictly true, we need
the probability distribution of the time-sharing sequence to be
symmetric, as it is the case in this paper.} Hence, symbols `1'
and `0' play a symmetric role, which strongly links the
conditional probabilities:
$\Prob{Y}{1|\Sigma=\sigma}=\Prob{Y}{0|\Sigma=c-\sigma}=1-\Prob{Y}{1|\Sigma=c-\sigma}$.
Therefore, Class-B collusion attacks are constrained in the
following way:
\begin{eqnarray}
    \boldtheta \in \class{B}{c} \triangleq
    \{\boldtheta: \theta_0=0, \theta_c=1,
    \theta_\sigma \in [0,1]\, \textrm{for}\, \sigma\in [c-1],\,
    \theta_\sigma=1-\theta_{c-\sigma},\, \textrm{for }
    \sigma\in[c-1]\}.
    \label{eq:classBcollusion}
\end{eqnarray}
Hence, a Class-B collusion attack has $\lfloor (c-1)/2 \rfloor$
degrees of freedom, and for even $c$ we necessarily have
$\theta_{c/2} = 1/2$. Clearly, for $c=2$ the only possible Class-B
collusion strategy is $\boldtheta = \{0, 0.5, 1\}$, which is also the Class-A attack.
Class-B collusion is relevant for traitor tracing in the multimedia
scenario, where each bit of the code is embedded in a different
chunk (frame, group of frames, etc.) of the multimedia signal by
means of a watermarking technique. The authors of~\cite{Amiri2009:High} refer to this class as ``strongly symmetric eligible channel.''

\paragraph{Class-C}
This is the classical collusion model used by cryptographers since
Boneh and Shaw~\cite{Boneh98:fingerprinting}.
The bits are directly pasted in the host content string, and thus the colluders
can compare their copies bitwise in order to disclose the location
of the traitor tracing code.
Class-C collusion attacks are no longer constrained like in Class
B, and new strategies are then possible such as the following:
\begin{itemize}
\item All `1's. The colluders put the symbol `1' whenever they can: $\theta_\sigma=1$, $0<\sigma\leq c$,
\item All `0's. The colluders put the symbol `0' whenever they can: $\theta_\sigma=0$, $0\leq \sigma< c$,
\end{itemize}
In general, a Class-C collusion strategy belongs to the following
set:
\begin{eqnarray}
    \boldtheta \in \class{C}{c}\triangleq \{\boldtheta: \theta_0=0, \theta_c=1, \theta_\sigma \in [0,1],\,
    \sigma \in [c]\},
    \label{eq:classCcollusion}
\end{eqnarray}
and therefore it has $c-1$ degrees of freedom.

\paragraph{Class-D}
This last class is quite special because it no longer fulfills the stationarity assumption introduced in \ref{subsub:MathModelColl}.
Now, the knowledge of the time-sharing sequence $\seq{p}{}$ is granted to the colluders.
From a statistical point of view,
the conditional probabilities  depend on $\Sigma_{i}$ and
$P_{i}$: $\Prob{Y_{i}}{1|\Sigma_{i},P_{i}}$. The collusion model for this class is a set of $c+1$ functions
$\boldtheta(p)=(\theta_0(p),\ldots,\theta_c(p))^T$ such that
\begin{eqnarray}
    \boldtheta(p) \in \class{C}{c},\,\forall p\in[0,1],
    \label{eq:classDcollusion}
\end{eqnarray}

The interest of Class-D is twofold: on one
hand, it gives the rate achievable by a code that does not
perform time-sharing (i.e. the value of $p$ is fixed) 
and on the other hand it shows the achievable rate when the code has been
broken (i.e. the secret of the probabilistic code has been
disclosed), meaning that the colluders know the value of $p_{i}$
for all index $i \in [m]$. Therefore, they can adapt their strategy for
each index chunk according to its value $p_{i}$. 
Notice that the attack is still assumed to be memoryless.

\subsection{Decoding families}	\label{sub:DecodingFamilies}
The study of traitor tracing codes from an achievable rate standpoint largely decouples their performances from any particular decoding algorithm. However, we consider two different families of decoders: the simple decoder \cite[Sec. 4]{Moulin08universal} and the joint decoder~\cite[Sec. 5]{Moulin08universal}. The simple decoder calculates the empirical mutual information between each user codeword and the pirated sequence, whereas the joint decoder calculates the empirical mutual information between each possible subset of $c$ users and the pirated sequence.
Due to their different nature, the two families have different achievable rates. Briefly, the joint decoder represents what the accusation side could do in an ideal world where complexity is not a matter, and it has been shown to be capacity-achieving. However, it has to tackle $\nchoosek{n}{c}$ groups which seems hardly affordable for large $n$. The simple decoder, suboptimal in general, represents the upper performance limit for more practical decoders.

\subsubsection{Joint decoder}
The achievable rate for the joint decoder against a given collusion attack is based on
the mutual information between $Y$, a symbol of the pirated sequence, and $X_{\set{C}}$, the symbols of the colluders' code sequences~\cite[Sec. 5]{Moulin08universal}. This holds for any index thanks to the symbol independence, and this is taken in expectation over the time-sharing random variable $P$: 
\begin{eqnarray}
R_{\joint}(\boldtheta)   & = & \frac{1}{c}\expect{P}{I(Y;X_{\set{C}}|P=p,\boldcaptheta=\boldtheta)} \nonumber\\
                & = & \frac{1}{c}\left(\expect{P}{H(Y|P=p,\boldcaptheta=\boldtheta)} -
    \expect{P}{H(Y|X_{\set{C}},P=p,\boldcaptheta=\boldtheta)}\right) \nonumber\\
    & = & \frac{1}{c}\left(\expect{P}{H(Y|P=p,\boldcaptheta=\boldtheta)} - \expect{P}{H(Y|\Sigma,P=p,\boldcaptheta=\boldtheta)}\right),
    \label{eq:rate-joint-decoder}
\end{eqnarray}
where $\Sigma$ is the random variable defined as the number of ones in the set $X_\set{C}$.  Equality in \eqref{eq:rate-joint-decoder} follows because of the assumption stated in Sect.~\ref{sub:Collusion}, namely that the output of the collusion channel only depends on the type of $X_\set{C}$, not on the order of its elements.
For the sake of clarity, we omit the expression $\boldcaptheta=\boldtheta$ in the sequel, but all the probability, entropy or mutual information expressions are given for a given collusion attack.
Plugging the collusion model introduced in~\ref{subsub:MathModelColl}, we have:
\begin{eqnarray}
\label{eq:ProbY}
\Prob{Y}{1|P=p}&=&\sum_{\sigma=0}^c\theta_\sigma\Prob{\Sigma}{\sigma|P=p},\\
\Prob{Y}{1|\Sigma=\sigma,P=p}&=&\theta_\sigma,
\end{eqnarray}
with $\Prob{\Sigma}{\sigma|P=p}=\nchoosek{c}{\sigma}p^\sigma(1-p)^{(c-\sigma)}$, known as the Bernstein polynomials~\cite{Sevy1995:Lagrange}.
Therefore, Eq.~\eqref{eq:rate-joint-decoder} can be rewritten as:
\begin{equation}
R_{\joint}(\boldtheta) =
\frac{1}{c}\left(\expect{P}{h_b\left(\sum_{\sigma=0}^c\theta_\sigma\Prob{\Sigma}{\sigma|P=p}\right)}
- \sum_{\sigma=0}^c \expect{P}{\Prob{\Sigma}{\sigma|P=p}}
h_b(\theta_\sigma) \right). \label{eq:RJoint}
\end{equation}

A possible interpretation of (\ref{eq:rate-joint-decoder}) is that the rate can also be expressed in terms of the average discrimination (or Kullback Leibler distance as~\cite{Blahut72}):
\begin{eqnarray}
    R_{\joint}(\boldtheta)&=& \expect{P}{D_{KL}(\Probb{Y,\Sigma}\,||\,\Probb{Y}\Probb{\Sigma}|P=p)}/c,\\
    &=&\expect{P}{\sum_{y\in\{0,1\}}\sum_{\sigma=0}^c\Prob{\Sigma}{\sigma|P=p}\Prob{Y}{y|\Sigma=\sigma}\log\left(\frac{\Prob{Y}{y|\Sigma=\sigma}}{\Prob{Y}{y|P=p}}\right)}/c.
    \label{eq:rjoint-kl}
\end{eqnarray}
The usefulness of this expression will become patent in Section~\ref{sub:JointC}.

\subsubsection{Simple decoder}
The achievable rate for the simple decoder against a given collusion attack is given in~\cite[Sec. 4]{Moulin08universal}:
\begin{eqnarray}
R_{\simple}(\boldtheta)  & = & \expect{P}{I(Y;X|P=p)} \nonumber\\
                & = & \expect{P}{H(Y|P=p)} -
    \expect{P}{H(Y|X,P=p)}\\
    &=&\expect{P}{D_{KL}(\Probb{X,Y}||\Probb{X}\Probb{Y}|P=p)}
    \label{eq:rate-simple-decoder}
\end{eqnarray}
This links the notion of rate to the inherent capability of distinguishing two hypothesis:
\begin{itemize}
	\item $\mathcal{H}_0$: User $j$ is innocent, and his codeword is independent of $Y$ given $P$: $\Prob{}{Y,X|\mathcal{H}_0}=\Prob{}{Y}\Prob{}{X}$, 
	\item $\mathcal{H}_1$: User $j$ is guilty and $Y$ has been created from his codeword: $\Prob{}{Y,X|\mathcal{H}_1}=\Prob{}{Y|X}\Prob{}{X}$.
\end{itemize}

The calculation of the rate needs the expressions of the
conditional probabilities induced by the collusion model:
\begin{eqnarray}
\Prob{Y}{1|X=1,P=p}&=&\sum_{k=1}^c\theta_k\nchoosek{c-1}{k-1}p^{k-1}(1-p)^{c-k},
\label{eq:CondProb1}\\
\Prob{Y}{1|X=0,P=p}&=&\sum_{k=0}^{c-1}\theta_k\nchoosek{c-1}{\phantom{ii}k}p^{k}(1-p)^{c-k-1}.
\label{eq:CondProb0}
\end{eqnarray}

\begin{proposition}
Two simple considerations:
\begin{itemize}
\item For the classes A, B and C, the following relationships hold:
\begin{eqnarray}
\Prob{Y}{1|X=1,P=p}&=&\Prob{Y}{1|P=p}+\frac{1-p}{c}\frac{\partial}{\partial p}\Prob{Y}{1|P=p},\\
\Prob{Y}{1|X=0,P=p}&=&\Prob{Y}{1|P=p}-\frac{p}{c}\frac{\partial}{\partial p}\Prob{Y}{1|P=p}.
\end{eqnarray}
\item For $c=2$, $\class{A}{2}=\class{B}{2}$.
\end{itemize}
\label{prop:CondProb}
\end{proposition}
\begin{proof}
After some manipulations, we have $\frac{\partial}{\partial p}\Prob{Y}{1|P=p}=c(\Prob{Y}{1|X=1,P=p}-\Prob{Y}{1|X=0,P=p})$.
Moreover, $\Prob{Y}{1|P=p}=p\Prob{Y}{1|X=1,P=p}+(1-p)\Prob{Y}{1|X=0,P=p}$. The second item is obvious.
\end{proof}

\subsubsection{Achievable rate under Class-Z attack}
\label{sec:class-z-attack}
For a given decoding family and size of collusion, the achievable rate of a code under Class-$Z$ attack (with $Z\in\{A,B,C,D\}$) is the mutual information produced by the worst collusion process in this class. For instance, with straightforward notation:
\begin{equation}
R_{\simple}^{Z}(c)=\min_{\boldtheta\in\class{Z}{c}} R_{\simple}(\boldtheta).
\label{eq:AchievRate}
\end{equation}
Since the colluders are more and more powerful as we consider upcoming classes, the following relationships hold for the simple decoder (and similarly for the joint decoder):
\begin{equation}
R_{\simple}^{D}(c)\leq R_{\simple}^{C}(c)\leq R_{\simple}^{B}(c)\leq R_{\simple}^{A}(c).
\label{eq:Inclusion}
\end{equation}

%

To stress the importance of the time-sharing variable $P$, it is interesting to define the function
\begin{equation}
	r_{\simple}^{Z}(c,p_0) \triangleq I(Y;X|P=p_0, \boldcaptheta = \boldtheta^*),
	\label{eq:rjoint-simple-def}
\end{equation}
where 
\begin{equation}
\boldtheta^* = \arg\min_{\boldtheta\in\class{Z}{c}} R_{\simple}(\boldtheta).
\nonumber
\end{equation}
The strong non-convexity of $r_{\simple}^{Z}(c,p)$ in $p$, in general, justifies the need of time-sharing~\cite{Moulin08universal}, as will be seen. 
Obviously, $R_{\simple}^{Z}(c)=\expect{P}{r^{Z}_{\simple}(c,p)}$. The extension of this definition to the joint decoder is straightforward.

\section{The joint decoder}
\label{sec:Joint}

\begin{figure*}[t] 
    \begin{minipage}[b]{1.0\textwidth}
        \centerline{\epsfig{figure = 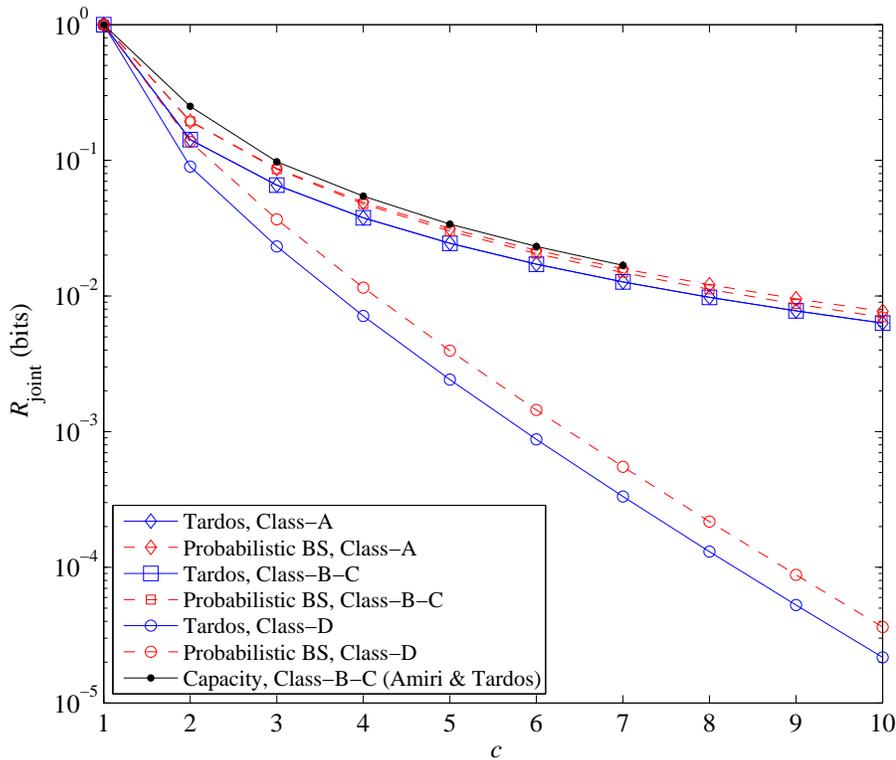, width=0.75\textwidth}}
    \end{minipage}
    \caption{Achievable rates for the joint decoder against different classes of collusion, for Tardos and probabilistic BS codes. The fingerprinting capacity (according to Amiri and Tardos) against Class-B-C colluders is plotted for comparison.}
    \label{fig:rates-joint}
\end{figure*}

\subsection{Colluders Class-A}
\label{sub:JointA}

\begin{proposition}
A Class-A collusion leads to pirated sequence symbols whose probability is $\Prob{Y}{1|P=p}=p$, for any collusion size.
Hence, the achievable rates of the joint decoder against Class-A collusion is
for the Tardos and the flat distribution:
\begin{equation}
R_{\joint}^{A}(c)=c^{-1}\left(\expect{P}{h_b(P)} - \sum_{\sigma=0}^c \expect{P}{\Prob{\Sigma}{\sigma|P}} h_b(\sigma/c)\right).
\label{eq:RateJointClassA}
\end{equation}
\end{proposition}
\begin{IEEEproof}
Since we have $\theta_\sigma=\sigma/c$, $\forall \sigma\in \{0,\ldots,c\}$, $\Prob{Y}{1|P=p}$ has a simple expression:
\begin{equation}
\Prob{Y}{1|P=p}=c^{-1}\expect{\Sigma}{\sigma|P=p}=p.
\label{eq:ProbYClassA}
\end{equation}
The last equality comes from the fact that $\Sigma$ is a random
variable distributed as a binomial $B(c,p)$, so its expectation is
$cp$.
\end{IEEEproof}
With the help of {\it Mathematica}, the expectations find closed form expressions, and the achievable rate in bits is for the Tardos pdf:
\begin{equation}
R_{\joint}^{A}(c)=c^{-1}\left(2-\log(e) - \pi^{-1}\sum_{\sigma=0}^c \frac{\Gamma(\sigma+1/2)\Gamma(c-\sigma+1/2)}{\Gamma(\sigma+1)\Gamma(c-\sigma+1)} h_b(\sigma/c)  \right)
\label{eq:TardosRateA}
\end{equation}
whereas for the flat pdf (i.e. for the probabilistic BS code):
\begin{equation}
R_{\joint}^{A}(c)=c^{-1}\left(\log(e)/2 - (c+1)^{-1}\sum_{\sigma=0}^c h_b(\sigma/c). \right)
\label{eq:FlatRateA}
\end{equation}

The resulting achievable rates are plotted in Fig.~\ref{fig:rates-joint}. These plots suggest that they decrease as $1/c^2$. This is confirmed by the next proposition.
\begin{proposition}
\label{prop:JointClassA}
For any pdf $f(p):[0,1]\rightarrow \mathbb{R}^+$, we have (in natural units)
\begin{equation}
\lim_{c\rightarrow+\infty}R_{\joint}^{A}(c)-\frac{1}{2\ln(2)c^{2}}=0.
\end{equation}
\end{proposition}
See Appendix~\ref{sub:ProofClassAJoint} for the proof.

Consider now the achievable rate as the expectation over $P$ of the function $r_{\joint}^A(c,p)$ defined according to \eqref{eq:rjoint-simple-def}. This function, which is plotted in Fig.~\ref{fig:rate-joint-class-a} for different values of $c$, is symmetric around $p=1/2$ because $h_b(\sigma/c)=h_b((c-\sigma)/c)$. For $c=2$ and $c=3$, its maximum is in $p=1/2$. This shows that the best pdf would be a Dirac's delta in $p=1/2$. There would no longer be need for time-sharing variable $P$, and the code $\set{X}$ would be composed of i.i.d. binary components with $\Prob{}{X_{ji}=1}=1/2$. For $c>3$, the maximum is no longer in $p=1/2$, but on two symmetric values depending on $c$, so that the capacity-achieving pdf is composed of two Dirac's deltas. This is a very special case where the capacity can be numerically derived. The achievable rates of the Tardos and the probabilistic BS codes are lower than the capacity, as can be seen in Fig.~\ref{fig:rates-joint}. The next section shows however that the Dirac's delta pdf achieving capacity in Class-A is a very dangerous choice under other collusion classes, and that time-sharing becomes a necessity.

\begin{figure*}[t] 
		\psfrag{prov}[c][]{$r_{\joint}^A(c,p)$ (bits)}
    \begin{minipage}[b]{1.0\textwidth}
        \centerline{\epsfig{figure = 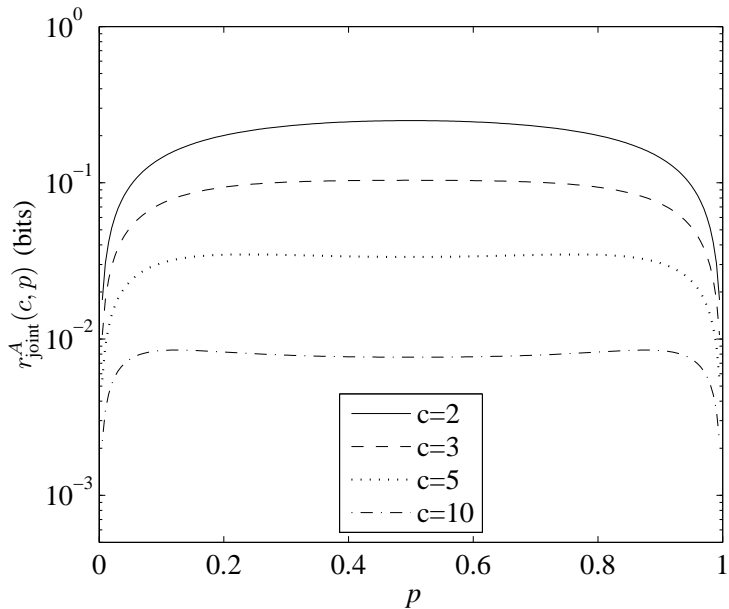, width=0.6\textwidth}}
    \end{minipage}
    \caption{Plot of $r_{\joint}^A(c,p)$ for the joint decoder under Class-A attack.}
    \label{fig:rate-joint-class-a}
\end{figure*}

\subsection{Colluders Classes B and C}	\label{sub:JointB}
\label{sub:JointC}

\subsubsection{2 colluders}
Thanks to Prop.~\ref{prop:CondProb}, the rationale for Class-A also holds for Class-B when $c=2$.
Tardos code rate is $R_{\joint}^{B}(2) = 7/8-\log(e)/2 \approx  0.154$ bits, whereas the rate for probabilistic BS code (i.e. flat pdf) $R_{\joint}^{B}(2) = \log(e)/4-1/6 \approx  0.194$ bits. The capacity is achieved with the Dirac's delta pdf, $f(p)=\delta(p-1/2)$, and it is approximately $0.25$ bits. We found back the same result as Amiri and Tardos (see line 2 of Table 1 in~\cite{Amiri2009:High}). However, this strategy is very risky if the number of colluders is actually bigger than 2 as we shall see in the next section.

\subsubsection{More colluders}
\label{sec:more-colluders-joint-class-b-c}
When $c\geq 3$, the analysis is much more complex and we have only succeeded to find out the worst collusion process and thus, the achievable rate for a given pdf $f(p)$.   

We resort now to the expression \eqref{eq:rjoint-kl} of the achievable rate in terms of the relative entropy.
The problem of minimizing \eqref{eq:rjoint-kl} can be rewritten as
a double minimization, exactly like the Blahut-Arimoto algorithm
for the computation of the rate-distortion function
\cite{Blahut72}. The main difference is that our minimization problem
corresponds to a degenerate rate-distortion problem where the only
distortion constraint is that $\boldtheta \in \class{B-C}{c}$ (in the sense that $\boldtheta \in \class{B}{c}$ or $\boldtheta \in \class{C}{c}$ depending on the class of colluders). The reader is referred to~\cite{Blahut72} or \cite[Chapt.~13]{CT91} for a detailed presentation of the Blahut-Arimoto algorithm as we only explain its application to our model.

In a slight abuse of notation, let us denote the rhs of (\ref{eq:rjoint-kl}) by $R(\Prob{Y}{Y|P},\boldtheta)$. 
The worst collusion process is disclosed by iteratively minimizing over each argument of this function, keeping the other constant. Thus, each iteration is comprised of two steps:

\begin{enumerate}
	\item		In the first step of the $k$-th iteration, for a fixed law $\Prob{Y}{1|P=p}=q^{(k-1)}(p)$ whose expression complies with~\eqref{eq:ProbY}, we minimize $R(q^{(k-1)}(p),\boldtheta)$ over $\boldtheta$. 
Note that $R(q^{(k-1)}(p), \boldtheta)$ is convex in $\boldtheta$ because for fixed $p_0 \in [0,1]$, or equivalently for fixed $\Prob{\Sigma}{\sigma|P=p_0}$, the argument of the expectation in (9) is convex in $\boldtheta$ \cite[Th.~2.7.4]{CT91}, and the expectation of this function over $P$ is still convex. Hence, the minimization of $R(q^{(k-1)}(p), \boldtheta)$ amounts to canceling the $(c-1)$ partial derivatives ($\theta_0$ and $\theta_c$ are already fixed to 0 and 1, respectively). Notice that we also have to impose the constraint $\boldtheta \in \class{B-C}{c}$. Ignoring temporarily this constraint, we have
\begin{equation}
    \frac{\partial}{\partial \theta_\sigma}R(q^{(k-1)}(p),\boldtheta)
     =  \frac{1}{c}\expect{P}{\Prob{\Sigma}{\sigma|P=p}\left(\log\frac{\theta_\sigma}
    {1-\theta_\sigma}+\log\frac{1-q^{(k-1)}(p)}{q^{(k-1)}(p)}\right)}.
\end{equation}
By setting the last expression to 0, we obtain
\begin{equation}
    \theta_\sigma^{(k)} = \frac{1}{1+B^{(k)}(\sigma)},\; \sigma=1,\ldots,c-1,
    \label{eq:1StepRJoint}
\end{equation}
with
\begin{equation}
B^{(k)}(\sigma)= 2\,\hat{}\,{\left(\frac{\expect{P}{\Prob{\Sigma}{\sigma|P=p}\log\frac{1-q^{(k-1)}(p)}{q^{(k-1)}(p)}}}{\expect{P}{\Prob{\Sigma}{\sigma|P=p}}}\right)}.
\label{eq:BSigma}
\end{equation}
Note that $B^{(k)}(\sigma)$ is well defined because: 
	\begin{itemize}
		\item	$q^{(k-1)}(p)$ is a polynomial of degree $c$ which equals $0$ only for $p=0$ (resp. $q^{(k-1)}(p)=1$ only for $p=1$);
		\item	$\Prob{\Sigma}{\sigma|P=p}$ is also a polynomial that goes to zero for $p=0$ and $p=1$. Hence, by continuity, the numerator of \eqref{eq:BSigma} equals 0 for $p=0$ and $p=1$;
		\item	the denominator of \eqref{eq:BSigma} doesn't cancel as there exist $p\in]0,1[$ such that $f(p)>0$.
	\end{itemize}
Finally, Eq.~(\ref{eq:1StepRJoint}) is always between 0 and 1, showing
that the constraint $\boldtheta \in \class{B-C}{c}$ is actually inactive.

	\item		The second step of the $k$-th iteration consists in updating the
function $\Prob{Y}{1|P=p}$ in order to provide the next function $q^{(k)}(p)$ with respect to the new collusion model $\boldtheta^{(k)}$ found in the first step. This is done by finding the function $q^{(k)}(p)$ minimizing the functional $R(q(p),\boldtheta^{(k)})$.
Let us denote by $r(q(p))$ the integrand of (\ref{eq:rjoint-kl}) for a fixed $\boldtheta^{(k)}$. We create an extension of the derivative $r(q(p))$ in $q(p)$ by a Taylor expansion of the difference $R(q(p)+\epsilon(p),\boldtheta^{(k)})-R(q(p),\boldtheta^{(k)})=\expect{P}{\left.\frac{\partial r}{\partial q}\right|_{q(p)}\epsilon(p)}+\expect{P}{o(\epsilon(p))}$. The minimum is reached for a function $q^{(k)}(p)$ such that any perturbation $\epsilon(p)$ doesn't change the value of the functional at least up to the first order. In other words, it cancels $\left.\frac{\partial r}{\partial q}\right|_{q(p)}$. This leads to the following update:
\begin{equation}
q^{(k)}(p)=\sum_{\sigma=0}^c\theta_\sigma^{(k)}\Prob{\Sigma}{\sigma|P=p}.
\label{eq:2StepRJoint}
\end{equation}

\end{enumerate}

Very much like for the Blahut-Arimoto algorithm, convergence to the worst collusion channel is monotonic, i.e. every step decreases the objective function. Since the optimization problem is convex, convergence to the optimal $\boldtheta$ is assured. 

Fig.~\ref{fig:rates-joint} shows the resulting achievable rate $R_{\joint}^{C}(c)$ when this algorithm is applied to the Tardos and probabilistic BS codes. We observe two surprising facts:

\begin{proposition} \label{prop:ClassBJoint1}
    For a symmetric $f(p)$ (being it a continuous pdf or a discrete pmf), the absolute minimum of the rate in $\boldtheta \in \class{C}{c}$ is achieved for a Class-B collusion attack (i.e. $\theta_\sigma^* = 1 - \theta_{c-\sigma}^*,\,\forall\sigma\in\{0,\ldots,c\}$), hence $R_{\joint}^{B}(c)=R_{\joint}^{C}(c)$.
\end{proposition}
\begin{proof}
See Appendix~\ref{sub:ProofPropClassBJoint1}.
\end{proof}
This proposition is contained in~\cite[Lemma 4.1]{Amiri2009:High} (although its proof is not given therein) which states that the capacity is achieved with a symmetric $f(p)$ and a ``strongly symmetric channel'', i.e. a Class-B attack.
Yet, these authors were able to `{\it computationally}' find out the worst collusion attack only  for the capacity-achieving pdf $f(p)$ whereas we provide here a powerful solving algorithm for any pdf. On the other hand, they are able to find the capacity-achieving pdf. Yet, it is a discrete pmf with strong dependency on the  collusion size, which is not known in practice. The worst case attacks are given in Table~\ref{tab:TardosjointB} for small collusion sizes and the Tardos pdf.

The second fact is illustrated in Figure~\ref{fig:rate-joint-class-b-c}(a) showing that the difference $\theta_\sigma^* - \sigma/c$ for the Tardos pdf is very small, especially for large $c$. This means that the optimal Class-B-C strategy for the Tardos pdf is surprisingly very close to the Class-A attack, a fact reflected in Figure~\ref{fig:rates-joint}, where one can see that the rates under Class-A-B-C attacks are indistinguishable for the Tardos pdf. Interestingly, it has been mentioned in \cite{Huang09} that the Class-A attack seems to be asymptotically optimal when the optimal $f(p)$ (which is asymptotically very close to the Tardos $f(p)$) is used. Nevertheless, this does not mean that $\theta_\sigma^*$ strictly converges to $\sigma/c$, neither for the Tardos $f(p)$ nor for other arbitrary time-sharing pdfs. As shown in Figure~\ref{fig:rate-joint-class-b-c}(b), for instance, the worst Class-B attack for the probabilistic BS code diverges from Class-A as $c$ is increased. Indeed, the achievable rates plotted in Figure~\ref{fig:rates-joint} for this code under Class A and Class-B-C are different. 

In conclusion, this section lowers the importance of the colluders classification introduced in Sec.~\ref{subsub:CollClass} for the Tardos and probabilistic BS codes. Surprisingly, there is no need to distinguish Class-B and Class-C, and experimentally, the worst collusion attack against the Tardos code seems to be very close to the Class-A. In the light of Prop.~\ref{prop:JointClassA}, this would mean that the achievable rate under Class-B of the Tardos code is asymptotically converging to $2\ln(2)/c^2$, just like the capacity for many pirates given in~\cite[Sec. 4.2]{Amiri2009:High}, which is also plotted in Fig.~\ref{fig:rates-joint} for reference. In this regard, it is worth recalling that the capacity-achieving time-sharing distribution depends on the number of colluders, whereas the Tardos and flat pdf remain independent of $c$.

\begin{figure*}[t] 
    \begin{minipage}[b]{0.49\textwidth}
    		\psfrag{prov}[c][]{$\theta_\sigma^* - \sigma/c$}
        \centerline{\epsfig{figure = 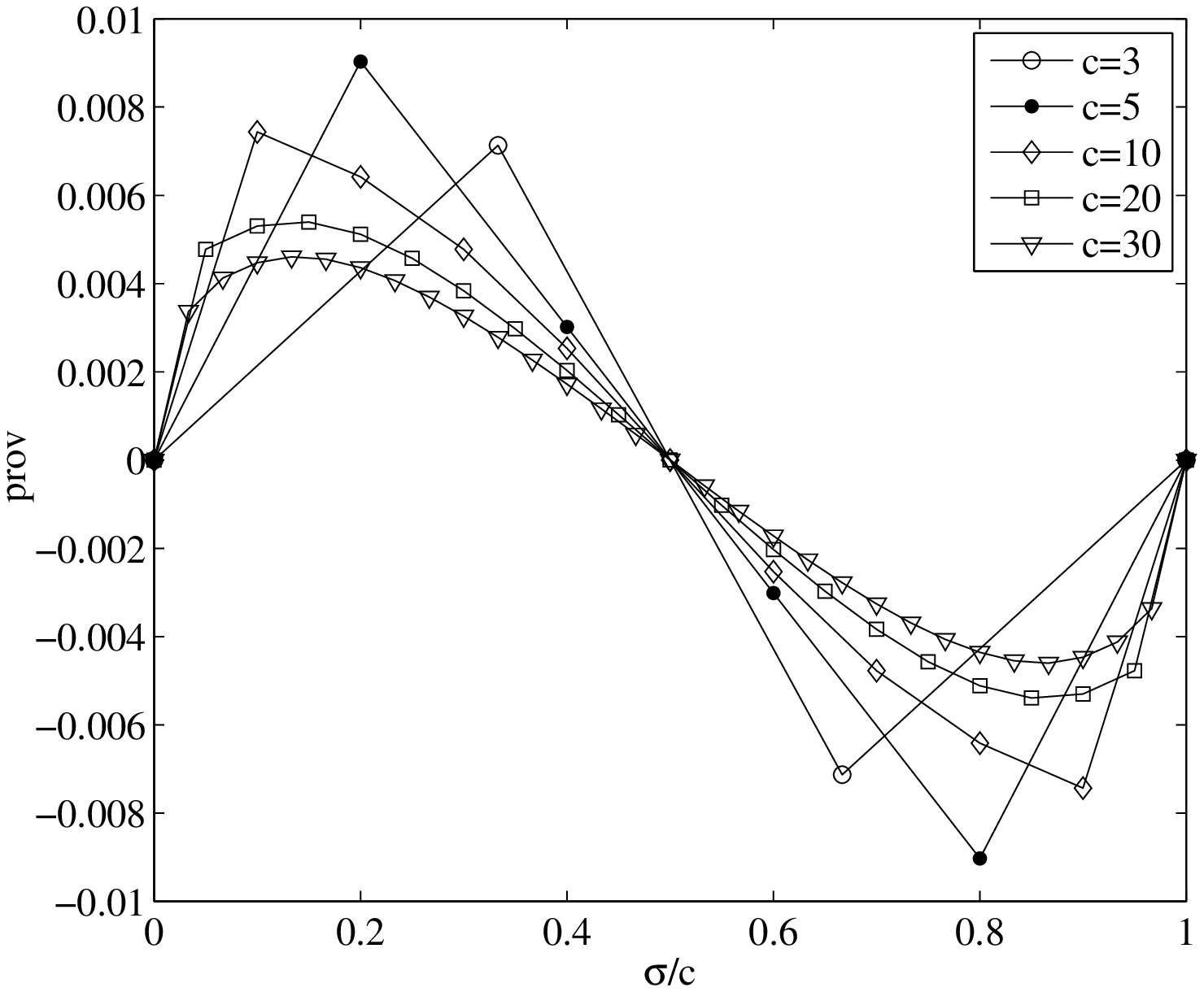,width=1.0\textwidth}}
        \centerline{(a) Tardos}
    \end{minipage}
    \begin{minipage}[b]{0.49\textwidth}
        \psfrag{prov}[c][]{$\theta_\sigma^* - \sigma/c$}
        \centerline{\epsfig{figure = 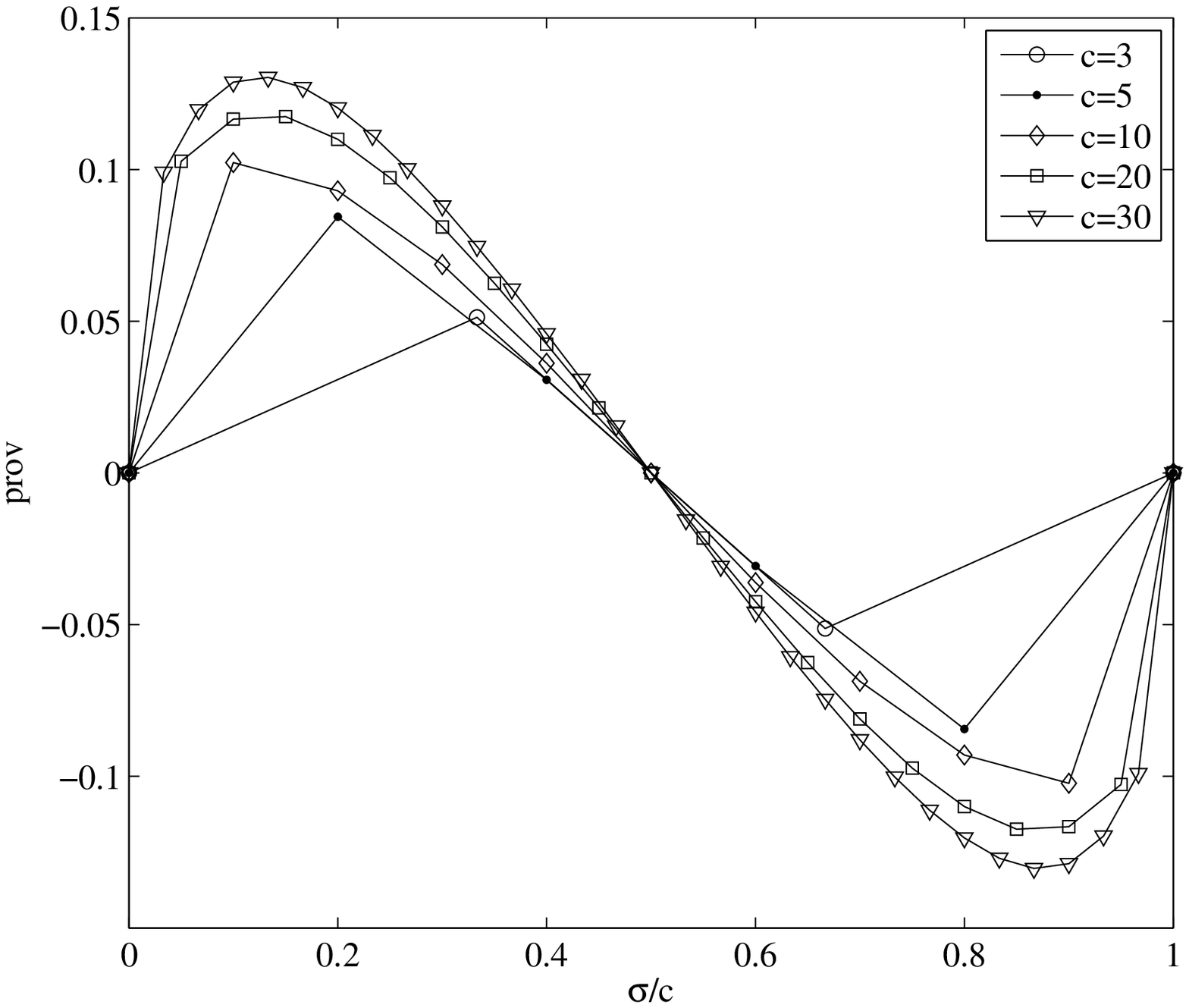, width=1.0\textwidth}}
        \centerline{(b) Probabilistic BS}
    \end{minipage}
    \caption{Worst Class-B-C collusion attack against the joint decoder.}
    \label{fig:rate-joint-class-b-c}
\end{figure*}

\begin{table}
\begin{center}
\caption{Worst collusion attacks, joint decoder, Tardos pdf, Class-C.}
\label{tab:TardosjointB}
  \begin{tabular}{|c|l|c|}
    \hline
    $c$ & $\boldtheta^*$ & $R_{\joint}^{B}(c)$ in bits\\
    \hline
    \hline
    2 & $(0,0.5,1)$ & 0.153 \\
    \hline
    3 & $(0,0.340,0.660,1)$ & 0.071 \\
    \hline
    4 & $(0,0.260,0.5,0.741,1)$ & 0.041 \\
    \hline
    5 & $(0,0.209,0.403,0.597,0.791,1)$ & 0.026 \\
    \hline
    6 & $(0,0.176,0.338,0.5,0.662,0.824,1)$ & 0.019 \\
    \hline
    7 & $(0,0.151, 0.291, 0.431,0.569,0.709, 0.849,1)$ & 0.014 \\
    \hline
    8 & $(0,0.133, 0.256, 0.378, 0.5, 0.622, 0.744, 0.867,1)$ & 0.011 \\
    \hline
    9 & $(0,0.119, 0.229, 0.338, 0.446, 0.554, 0.662, 0.771, 0.881,1)$ & 0.008 \\ \hline 
  \end{tabular}
  \end{center}
\end{table}

\subsection{Colluders Class-D}	\label{sub:JointD}
Remember that Class D colluders have disclosed
the exact values of $\{p_i\}_{i\in[m]}$, so that their strategy
is, a priori, no longer stationary, but on the contrary dependent
on $p$. The colluders minimize the achievable rate of the code by finding
the worst collusion attack $\boldtheta^*(p_i)$ minimizing
$I(Y;\Sigma|P=p_i)$.

\begin{proposition}\label{prop:jointD}
The worst case collusion strategy minimizing the rate of the joint decoder is given by $\boldtheta^*(p) = [0, \theta^*(p),\ldots,\theta^*(p),1]$, with
\begin{eqnarray}
    \theta^*(p) = \frac{p^c}{p^c+(1-p)^c}.
    \label{eq:optimal-attack-joint-knownp-final}
\end{eqnarray}
\label{th:worstattack-joint-knownp}
\end{proposition}
\begin{IEEEproof}
    See Appendix~\ref{app:worstattack-joint-knownp}.
\end{IEEEproof}

The worst case attack is not constant along the pirated sequence
if time-sharing has been done. Note also that it depends on $c$,
the number of colluders. Interestingly, as $c$ grows, the worst
case attack amounts to a simple deterministic strategy, which
depends only on whether $p$ is larger or smaller than $1/2$, as
illustrated in Fig.~\ref{fig:opt-joint-knownp}(a). It summarizes
as selecting the All `1' (resp. All `0') strategy when $p>1/2$ (resp.
$p<1/2$) and the `coin-flip' strategy (cf. Sect.~\ref{subsub:CollClass}) when $p=1/2$. The resulting $r_{\joint}^{D}(c,p)$ is shown in
Fig.~\ref{fig:opt-joint-knownp}(b) for different values of $c$ as
a function of $p$. Fig.~\ref{fig:rates-joint} shows the
achievable rate $R_{\joint}^D(c)=\expect{P}{r_{\joint}^{D}(c,p)}$ for Tardos and probabilistic BS codes, where we can see that both rapidly decrease with similar slope. It is very interesting to notice that, although the colluders have disclosed the secret of the code, they cannot set the rate to 0. Nevertheless, the following proposition shows that the capacity vanishes exponentially fast as the number of colluders is increased.

\begin{proposition}
\label{th:joint_capacity_classD}
For any value of $c\geq 2$, capacity under Class-D collusion is achieved with $f(p)=\delta(p-1/2)$, and it is given in bits by 
\begin{equation}
C^{D}(c) = \frac{1}{c2^{c-1}}.
\label{eq:CapaJointD}
\end{equation} 
\end{proposition}
\begin{IEEEproof}
See Appendix~\ref{app:joint_capacity_classD}   
\end{IEEEproof}
According to~\cite{Moulin08universal}, \eqref{eq:CapaJointD} is also the capacity reached by a code which does not perform time-sharing. Thus, as far as a joint decoder is concerned, time-sharing under Class-D collusion does not bring any gain in terms of capacity.
The comparison between this exponentially vanishing capacity (with exponent $1-c$) and the achievable rate under Class-A, which is only decreasing in $1/c^2$, illustrates the dramatic benefits of keeping secret the time-sharing sequence.

\begin{figure*}[t] 
    \begin{minipage}[b]{0.49\textwidth}
        \centerline{\epsfig{figure = 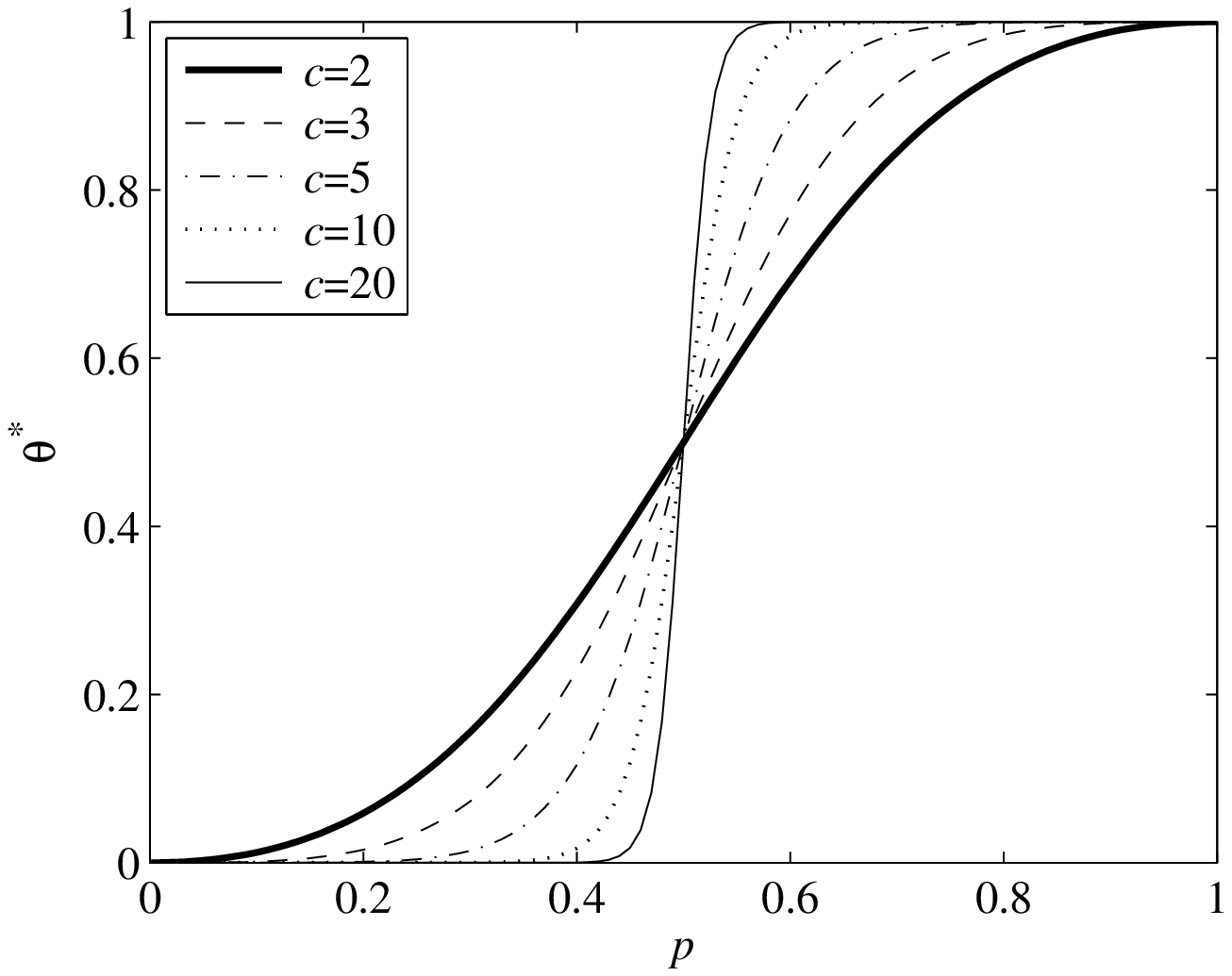, width=1.0\textwidth}}
        \centerline{(a)}
    \end{minipage}
    \begin{minipage}[b]{0.49\textwidth}
    \psfrag{prov}[c][]{$r_{\joint}^{D}(c,p)$ (bits)}    		
        \centerline{\epsfig{figure = 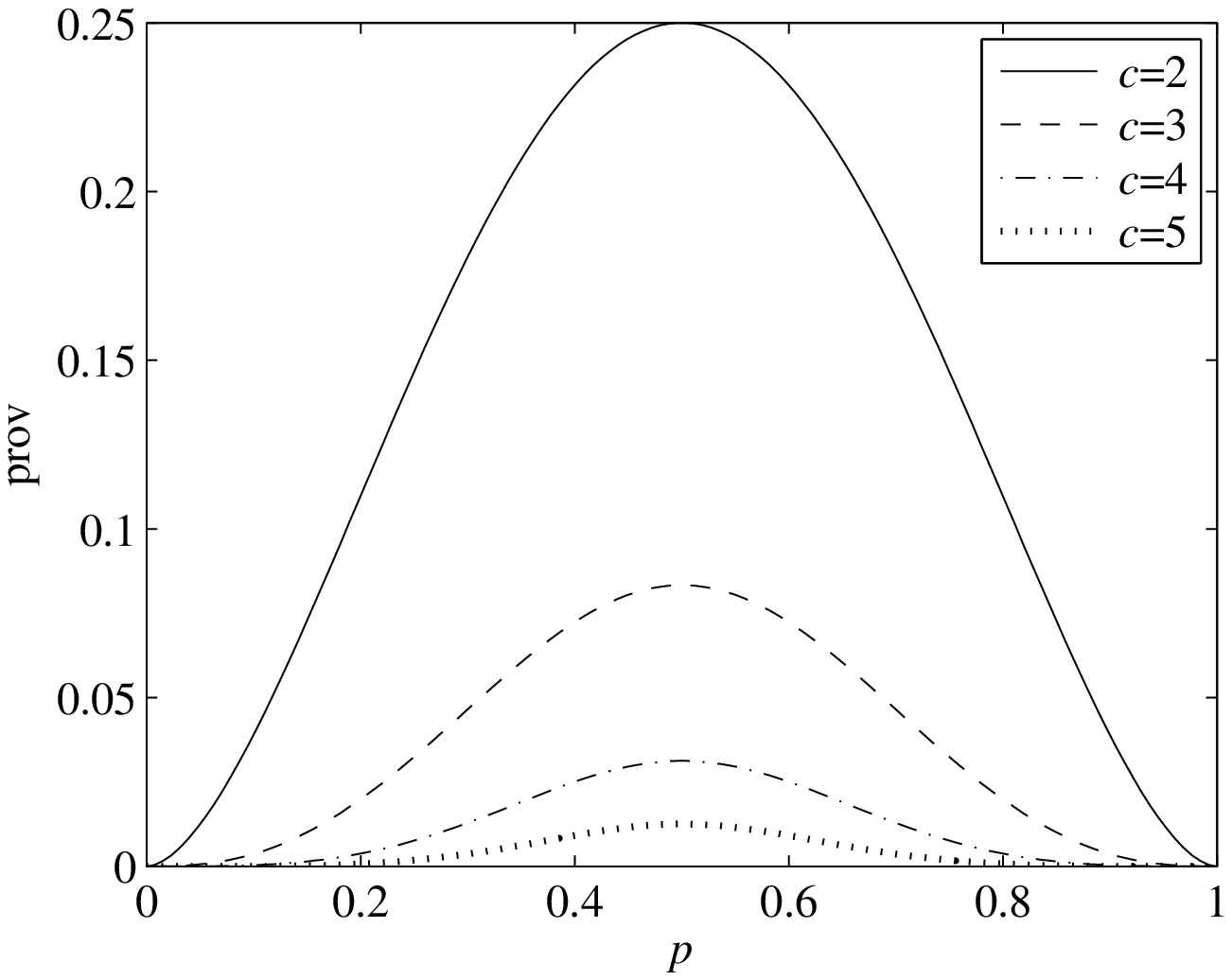, width=1.0\textwidth}}
        \centerline{(b)}
    \end{minipage}
    \caption{Worst case Class-D attack against the joint decoder: parameter of the worst case attack (a),
    and $r_{\joint}^{D}(c,p)$ (b).}
    \label{fig:opt-joint-knownp}
\end{figure*}

\section{The simple decoder}
\label{sec:simple}

\begin{figure*}[t] 
    \begin{minipage}[b]{1.0\textwidth}
        \centerline{\epsfig{figure = 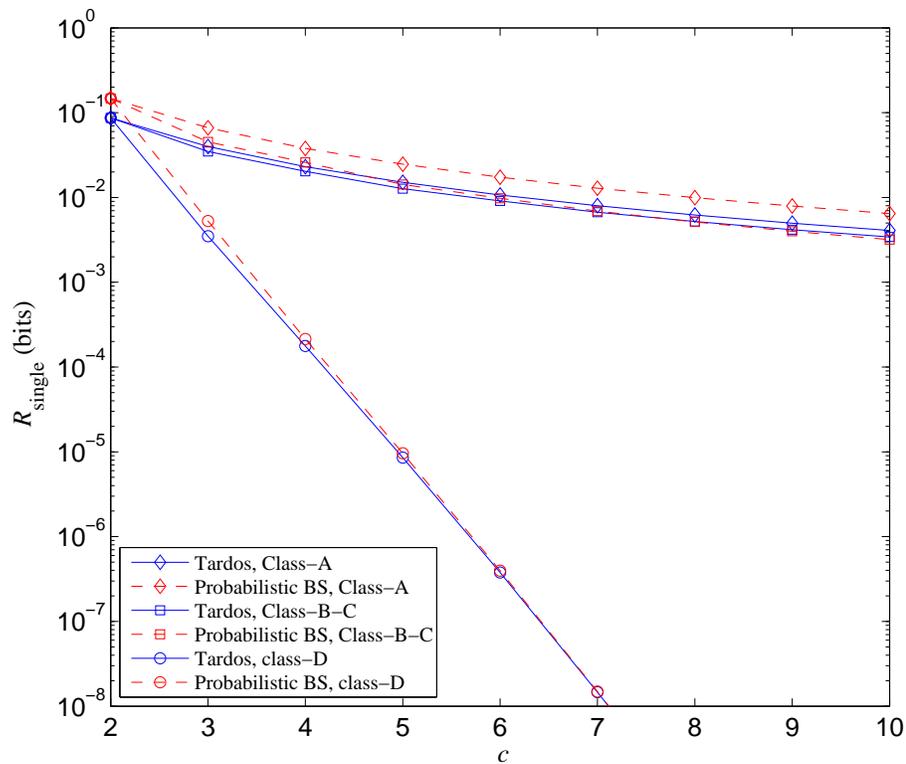, width=0.75\textwidth}}
    \end{minipage}
    \caption{Achievable rates for the simple decoder against different classes of collusion, for Tardos and probabilistic BS codes.}
    \label{fig:rate-simple}
\end{figure*}

\subsection{Colluders Class-A}
\label{sub:simpleA}

\begin{proposition}
A Class-A collusion produces the following achievable rate:
\begin{equation}
R_{\simple}^{A}(c)=\expect{P}{h_b(P)}-\expect{P}{Ph_b(P+(1-P)/c)+(1-P)h_b(P(1-1/c))}
\end{equation}
\end{proposition}
\begin{IEEEproof}
Prop.~\ref{prop:CondProb} and \eqref{eq:ProbYClassA} yield that $\Prob{Y}{1|X=1,P=p}=p+(1-p)/c$ and $\Prob{Y}{1|X=0,P=p}=p-p/c$. These expressions are then plugged in~\eqref{eq:rate-simple-decoder}.    
\end{IEEEproof}

Fig.~\ref{fig:rate-simple} shows the achievable rate (obtained through numerical integration) against the collusion size for the Tardos and the probabilistic BS codes. As can be seen, the rate for the probabilistic BS code against Class-A is higher than that of the Tardos code. The reason is in the shape of $r_{\simple}^A(c,p)$, which is plotted in Fig.~\ref{fig:rate-simple-class-a} for different values of $c$. This figure suggests that $r_{\simple}^A(c,p)$ achieves its global maximum at $p=1/2$. According to this, capacity for the simple decoder against Class-A would be achieved when $f(p) = \delta(p - 1/2)$.

\begin{figure*}[t] 
		\psfrag{prov}[c][]{$r_{\simple}^A(c,p)$ (bits)}
    \begin{minipage}[b]{1.0\textwidth}
        \centerline{\epsfig{figure = 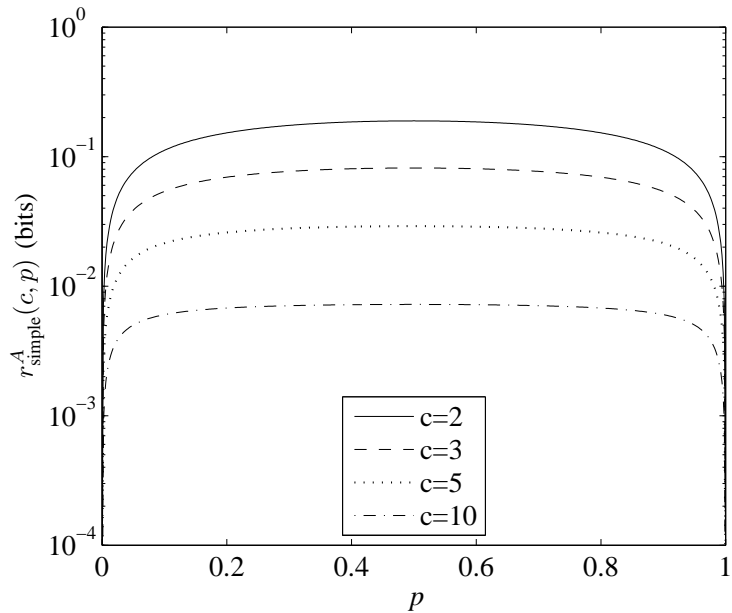, width=0.6\textwidth}}
    \end{minipage}
    \caption{Plot of $r_{\simple}^A(c,p)$ for the simple decoder under Class-A attack.}
    \label{fig:rate-simple-class-a}
\end{figure*}

\subsection{Colluders Classes B and C}
\label{sub:simpleB}

We do not have any proof for the Classes B and C. We were only able to find the worst collusion attack thanks to a numerical optimization tool which performs well only if the collusion size is not too big: $c\leq 15$. This was done for the Tardos and probabilistic BS codes. For the Tardos code, Table~\ref{tab:TardossimpleB} shows the resulting worst collusion attacks for $c<10$. The rate achievable by Tardos and probabilistic BS codes under the worst Class B-C attacks is plotted in Fig.~\ref{fig:rate-simple}, which suggests that the Tardos pdf is a better choice than the flat pdf as the number of colluders increases (note that the rate for the flat pdf already becomes lower than that of the Tardos pdf for $c=8$).

The observation of the numerical results allows us to formulate the two following conjectures (without formal proof).

\begin{conjecture} \label{conj:ClassBsimple1}
    For a symmetric $f(p)$, the Class-C worst collusion attack indeed belongs to the Class-B subset, i.e. $\theta_\sigma^* = 1 - \theta_{c-\sigma}^*,\,\forall\sigma\in[c]$, and $R_{\simple}^{B}(c)=R_{\simple}^{C}(c)$.
\end{conjecture}

\begin{conjecture} \label{conj:ClassBsimple2}
	For the Tardos pdf $f(p)=\sqrt{p(1-p)}/\pi$, the worst collusion attack surprisingly makes the probability $\Prob{Y}{1|P=p}$ converging to $q^{\mbox{\scriptsize conv}}(p)=2\arcsin(\sqrt{p})/\pi$, as the collusion size increases.\footnote{Note that $q^{\mbox{\scriptsize conv}}(p_0)$ is nothing but the integral in $p$ of the Tardos pdf from $0$ to $p_0$.} More specifically, $\Prob{Y}{1|P=p}$ is the orthogonal projection of $q^{\mbox{\scriptsize conv}}(p)$ over the affine subspace spanned by the Bernstein polynomials $\{\Prob{\Sigma}{\sigma|P=p}\}_{\sigma\in[c-1]}$ and containing the polynomial $\Prob{\Sigma}{c|P=p}$. In other words, $\int_{0}^1 (\Prob{Y}{1|P=p}-q^{\mbox{\scriptsize conv}}(p))\Prob{\Sigma}{\sigma|P=p}dp=0,\,\forall\sigma\in[c-1]$.
\end{conjecture}

Fig.~\ref{fig:rate-simple-unknownp}(a) illustrates how the probability $\Prob{Y}{1|P=p}$ quickly converges to $q^{\mbox{\scriptsize conv}}(p)$ (the thick solid line) as $c$ is increased. Fig.~\ref{fig:rate-simple-unknownp}(b) shows the resulting rates $r_{\simple}^C(c,p)$.
According to the second conjecture, in practice we can obtain the parameters of the worst case attack for the Tardos pdf by performing the projection of $q^{\mbox{\scriptsize conv}}(p)-\Prob{\Sigma}{c|P=p}$ onto the linear subspace spanned by the Bernstein polynomials. The Durrmeyer-Sevy algorithm is an elegant way to perform this orthogonal projection~\cite[Th. 2]{Sevy1995:Lagrange}.  

\begin{figure}[t]
    \begin{minipage}[b]{0.49\textwidth}
        \psfrag{prov}[c][]{{\scriptsize $\Prob{Y}{1|P=p}$}}
        \centerline{\epsfig{figure = 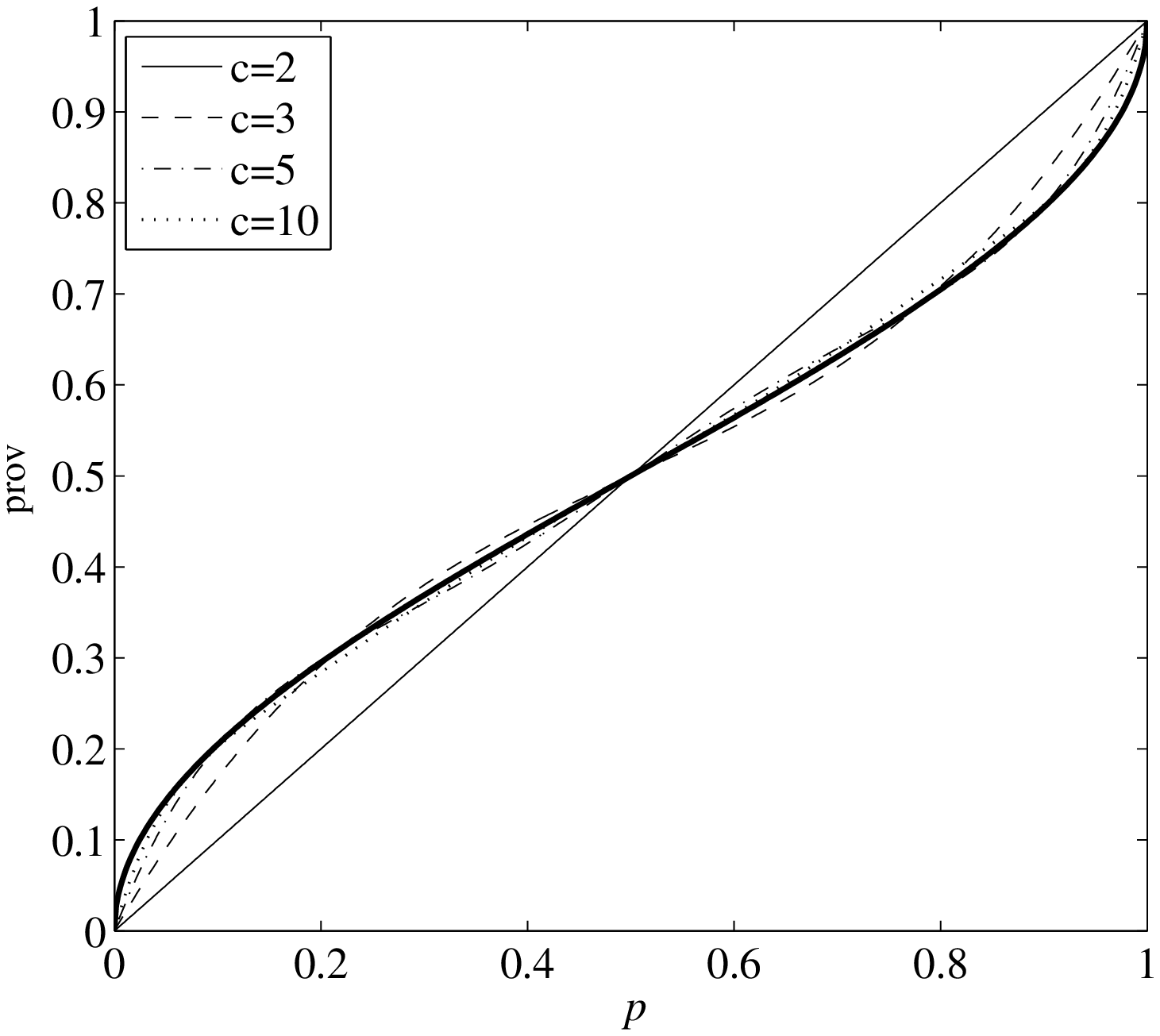, width=0.98\textwidth}}
        \centerline{(a)}
    \end{minipage}
    \begin{minipage}[b]{0.49\textwidth}
    		\psfrag{prov}[c][]{$r_{\simple}^B(c,p)$ (bits)}
        \centerline{\epsfig{figure = 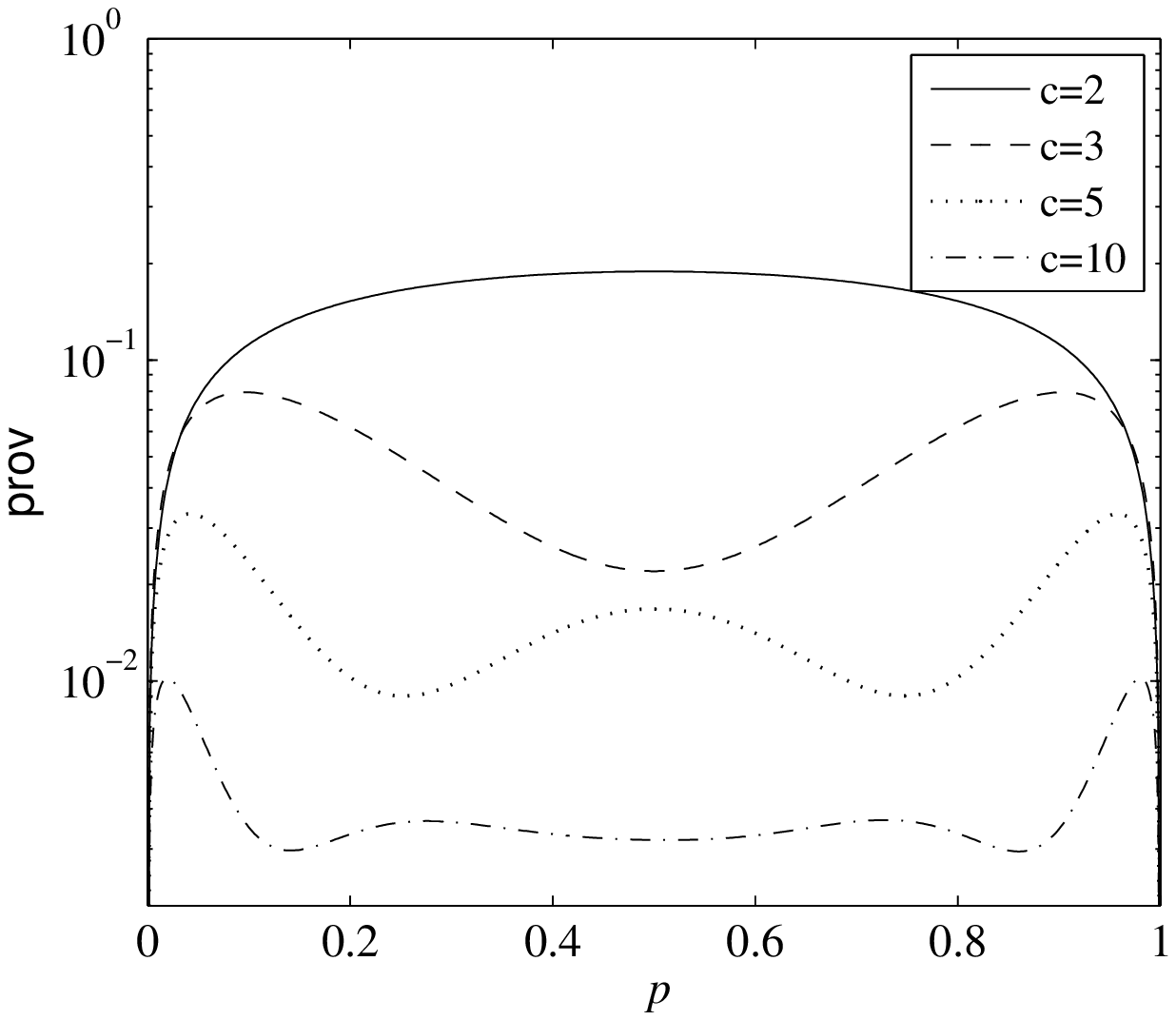, width=1.0\textwidth}}
        \centerline{(b)}
    \end{minipage}    
    \caption{Worst Class-C collusion attack against the simple decoder and Tardos time-sharing pdf: plot of $\Prob{Y}{1|P=p}$ (a), and plot of $r_{\simple}^C(c,p)$ (b).}
    \label{fig:rate-simple-unknownp}
\end{figure}

\begin{table}
\begin{center}
\caption{Worst collusion attacks, simple decoder, Tardos pdf, Class-C.}
\label{tab:TardossimpleB}
  \begin{tabular}{|c|l|c|}
    \hline
    $c$ & $\boldtheta^*$ & $R_{\simple}^{B}(c)$ in bits\\
    \hline
    \hline
    2 & $(0,0.5,1)$ & 0.087 \\
    \hline
    3 & $(0,0.652,0.348,1)$ & 0.035 \\
    \hline
    4 & $(0,0.488,0.5,0.512,1)$ & 0.02 \\
    \hline
    5 & $(0,0.594,0.000,1.000,0.406,1)$ & 0.013 \\
    \hline
    6 & $(0,0.503,0.175,0.500,0.825,0.497,1)$ & 0.009 \\
    \hline
    7 & $(0,0.492, 0.000, 0.899,0.101,1.000, 0.508,1)$ & 0.007 \\
    \hline
    8 & $(0,0.471, 0.000, 0.689, 0.500, 0.310, 1.000, 0.529,1)$ & 0.005 \\
    \hline
    9 & $(0,0.440, 0.000, 0.698, 0.230, 0.770, 0.302, 1.000, 0.560,1)$ & 0.004 \\ \hline 
  \end{tabular}
  \end{center}
\end{table}

\subsection{Colluders Class-D}
\label{sub:simpleD}
The mutual information between $Y$ and $X$ knowing the value of $p$ is as follows:
\begin{eqnarray}
I(Y;X|P=p)&=&H(Y|P=p)-H(Y|X,P=p)\\
&=& h_b(\Prob{Y}{1|P=p})-ph_b(\Prob{Y}{1|X=1,P=p}) \nonumber\\
& - & (1-p)h_b(\Prob{Y}{1|X=0,P=p})
\end{eqnarray}
\subsubsection{2 colluders}
For the case $c=2$, the collusion strategy has only one degree of freedom, i.e. $\boldtheta= [0, \theta_1, 1]$.
\begin{proposition}
    The worst collusion strategy for $c=2$ is given by $\theta_1^* = p^2/(p^2+(1-p)^2)$,
    exactly as for the joint decoder (see Prop.~\ref{prop:jointD}).
\end{proposition}
\begin{IEEEproof}
    The steps are roughly the same as those followed in
    Appendix~\ref{app:worstattack-joint-knownp} for the joint decoder.
    Taking the derivative of (\ref{eq:rate-simple-decoder}) with respect to
    $\theta_\sigma$ we obtain
    \begin{eqnarray}
        \partialder{\theta_\sigma} I(Y;X|P=p) = \Prob{\Sigma}{\sigma|P=p}
        \log\left(A(\boldtheta, p)\right),
    \end{eqnarray}
    where
    \begin{eqnarray}
        A(\boldtheta, p) =
        \left(\frac{1-\Prob{Y}{1|P=p}}{\Prob{Y}{1|P=p}}\right)
        \left(\frac{\Prob{Y}{1|X=1,P=p}}{1-\Prob{Y}{1|X=1,P=p}}\right)^{\sigma/c}
        \left(\frac{\Prob{Y}{1|X=0,P=p}}{1-\Prob{Y}{1|X=0,P=p}}\right)^{(c-\sigma)/c}.
    \end{eqnarray}
    It only remains to search for the collusion strategy that makes $A(\boldtheta,p)=1$,
    taking into account that for $c=2$, $\boldtheta=[0, \theta_1,
    1]$.
\end{IEEEproof}

\subsubsection{More colluders}
When $c>2$, obtaining a closed-form expression for the worst case
attack is not possible, in general. However, it is possible to
reduce the computation of the optimal collusion strategy to
solving for a simple line search or linear equation. This is based
on some fundamental results given in
Lemma~\ref{lemma:null-rate-simpledec} and
Lemma~\ref{lemma:opt-strat-simpledec-outrange} below.


\begin{lemma}
    The worst case collusion strategy when 3 or more colluders are involved achieves
    null rate in the range $p \in [\eta_c, 1-\eta_c]$, with $\eta_{c}$ the unique real root in the interval $[1/c,2/c]$ of the following polynomial: 
\begin{equation}
    (1-p)^{c-2}(1 -cp) + p^{c-1}.
    \label{eq:polynomial-null-rate}
\end{equation}
Moreover, the value of $\eta_c$ asymptotically approaches $1/c$ as $c$ is increased.
\label{lemma:null-rate-simpledec}
\end{lemma}
\begin{IEEEproof}
See Appendix~\ref{app:null-rate-simple-knownp}.
\end{IEEEproof}


\begin{lemma}
Let $\eta_c$ be the root given in
Lemma~\ref{lemma:null-rate-simpledec}. For $p \notin [\eta_c,
1-\eta_c]$ and $c\geq 3$, there is at most one component of $\boldtheta^*(p)$ which is not equal to zero or one:
\begin{itemize}
    \item If $p<\eta_c$, the worst collusion is of the form $\boldtheta_a(p) =(0,\theta_1(p),0,\ldots,0,1)^T$. Furthermore, $\theta_1(p) = 1$ for $p\in[1/c,\eta_c]$.
    \item If $p>1-\eta_c$, the worst collusion is of the form
    $\boldtheta_b(p) =(0,1,\ldots,1,\theta_{c-1}(p),1)^T$. Furthermore, $\theta_{c-1}(p) = 0$ for $p\in[1-\eta_c,1/c]$.
\end{itemize}
\label{lemma:opt-strat-simpledec-outrange}
\end{lemma}
\begin{IEEEproof}
    See Appendix~\ref{app:LemmaClassDsimple}.
\end{IEEEproof}
Using lemmas~\ref{lemma:null-rate-simpledec} and
\ref{lemma:opt-strat-simpledec-outrange}, the optimal collusion
strategy is characterized by the following proposition.
\begin{proposition}
    The worst Class-D collusion strategy $\boldtheta^*(p)$ for a simple decoder is given by:
    \begin{enumerate}
        \item   In the interval $p\in[\eta_c,1-\eta_c]$,
        $\boldtheta^*(p) \in \mathcal{H}^c$,
        where $$\mathcal{H}^c \triangleq \{\boldtheta \in \mathcal{P}_D(c):
        \boldtheta^T (\bm q_{\Sigma1} - \bm q_{\Sigma0}) = 0\},$$
        with
        \begin{eqnarray}
	        {\bm q}_{\Sigma1}&=&(\Prob{\Sigma}{0|X=1,P=p},\ldots,\Prob{\Sigma}{c|X=1,P=p})^T \label{eq:definition_qsigma1}\\
	        {\bm q}_{\Sigma0}&=&(\Prob{\Sigma}{0|X=0,P=p},\ldots,\Prob{\Sigma}{c|X=0,P=p})^T.
	        \label{eq:definition_qsigma0}
        \end{eqnarray}

        \item   For $p \notin [\eta_c,1-\eta_c]$, $\boldtheta^*(p)$
        is given by
        Lemma~\ref{lemma:opt-strat-simpledec-outrange} with $\theta_1(p)
        = 1 - \theta_{c-1}(1-p) = \theta^*$, which is defined as
        \begin{eqnarray}
            \theta^* \triangleq \arg\min_{\theta} \left( h_{b}(g_1(\theta,c,p)) -
            p h_{b}(g_2(\theta,c,p)) -
            (1-p)h_{b}(g_3(\theta,c,p)) \right),
            \label{eq:linesearch}
        \end{eqnarray}
        where
        \begin{eqnarray}
            g_1(\theta,c,p) & = & \theta c p(1-p)^{c-1} + p^c,
            \nonumber\\
            g_2(\theta,c,p) & = & \theta(1-p)^{c-1}+p^{c-1}, \nonumber\\
            g_3(\theta,c,p) & = & \theta(c-1)p(1-p)^{c-2}.
            \nonumber
        \end{eqnarray}
    \end{enumerate}
    \label{prop:optimal-classD-simple}
\end{proposition}
\begin{IEEEproof}
    Proving the first part of the proposition is straightforward: if $p$ belongs to the
    interval defined in
    Lemma~\ref{lemma:null-rate-simpledec}, then it necessarily implies
    that the global minimum of the mutual information functional is
    achieved by a vector $\boldtheta \in \mathcal{P}^D(c)$. In such
    case, according to the proof of
    Lemma~\ref{lemma:null-rate-simpledec},
    the optimal collusion strategy must fulfill the condition
    (\ref{eq:condition-null-mutinf}).

    For proving the second part of the proposition we resort to
    Lemma~\ref{lemma:opt-strat-simpledec-outrange}, which states
    that for $p \notin [\eta_c, 1-\eta_c]$ the optimal strategy
    has only one degree of freedom. We have to consider two cases:
    \begin{enumerate}
        \item   If $p<1/2$: We have $\Prob{Y}{1|P=p,\boldcaptheta=\boldtheta_a(p)} =
        g_1(\theta_1(p), c, p)$, $\Prob{Y}{1|P=p,\boldcaptheta=\boldtheta_a(p),X=1} =
        g_2(\theta_1(p), c, p)$, and $\Prob{Y}{1|P=p,\boldcaptheta=\boldtheta_a(p),X=0} =
        g_3(\theta_1(p), c, p)$. Hence, the parameter $\theta_1(p)$ of the optimal collusion strategy is the result of (\ref{eq:linesearch}).

        \item   If $p>1/2$:
        $\Prob{Y}{1|P=p,\boldcaptheta=\boldtheta_b(p)} =
        1 - g_1(1-\theta_{c-1}(p), c, 1-p)$, $\Prob{Y}{1|P=p,\boldcaptheta=\boldtheta_b(p),X=1} =
        1 - g_3(1-\theta_{c-1}(p), c, 1-p)$, and $\Prob{Y}{1|P=p,\boldcaptheta=\boldtheta_b(p),X=0} =
        1 - g_2(1-\theta_{c-1}(p), c, 1-p)$. Taking into account the symmetry of the binary entropy function $h_{b}(.)$, it is easy to see
        that for the optimum strategy $\theta_{c-1}(p)=1-\theta_{1}(1-p)$.
    \end{enumerate}
\end{IEEEproof}

\begin{figure*}[t] 
    \begin{minipage}[b]{0.49\textwidth}
    		\psfrag{prov}[c][]{{\footnotesize $\boldtheta^*$}}
        \centerline{\epsfig{figure = 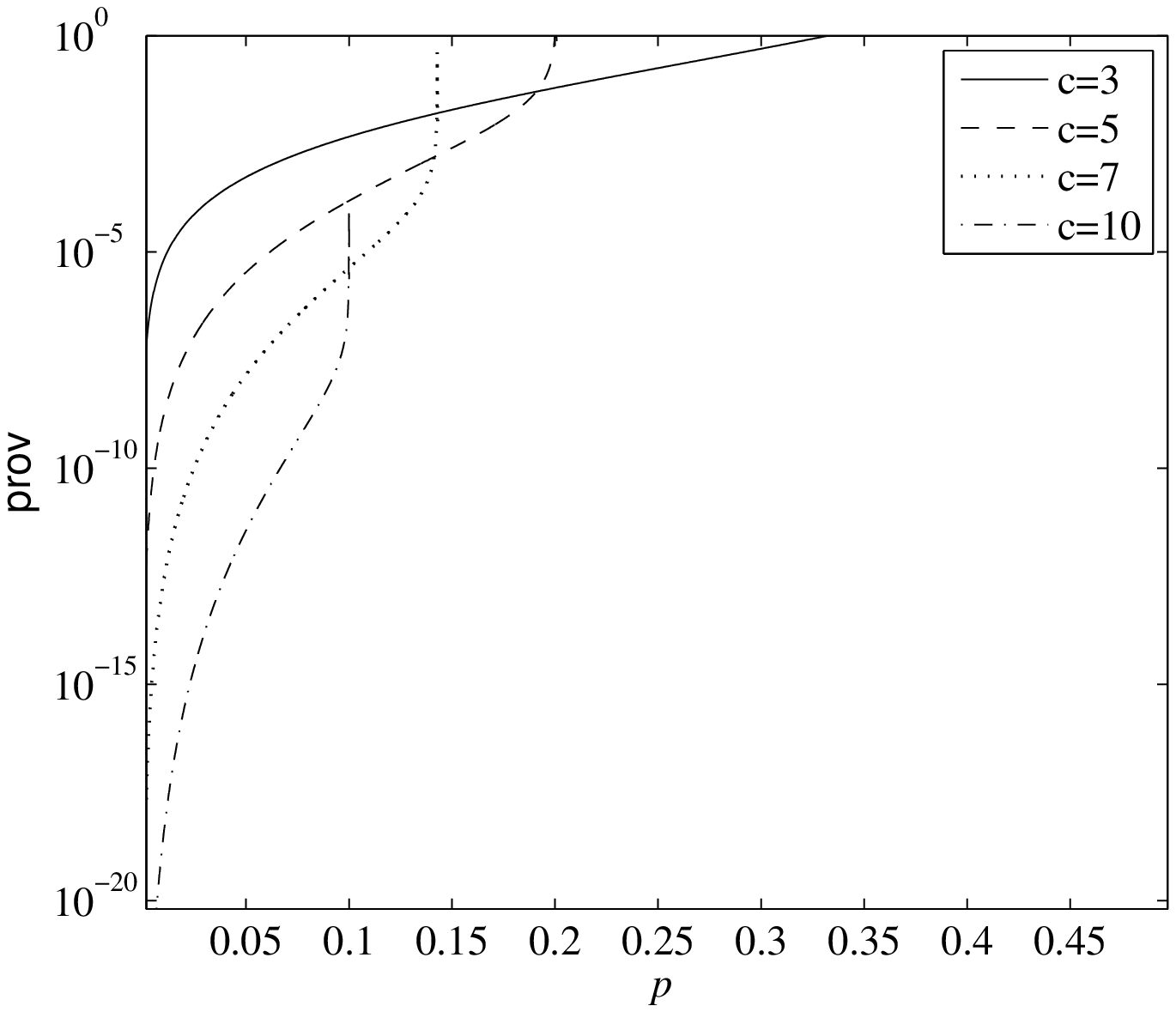, width=1.0\textwidth}}    
        \centerline{(a)}
    \end{minipage}
    \begin{minipage}[b]{0.49\textwidth}
    		\psfrag{prov}[c][]{{\footnotesize $r_{\simple}^D(c,p)$ (bits)}}
        \centerline{\epsfig{figure = 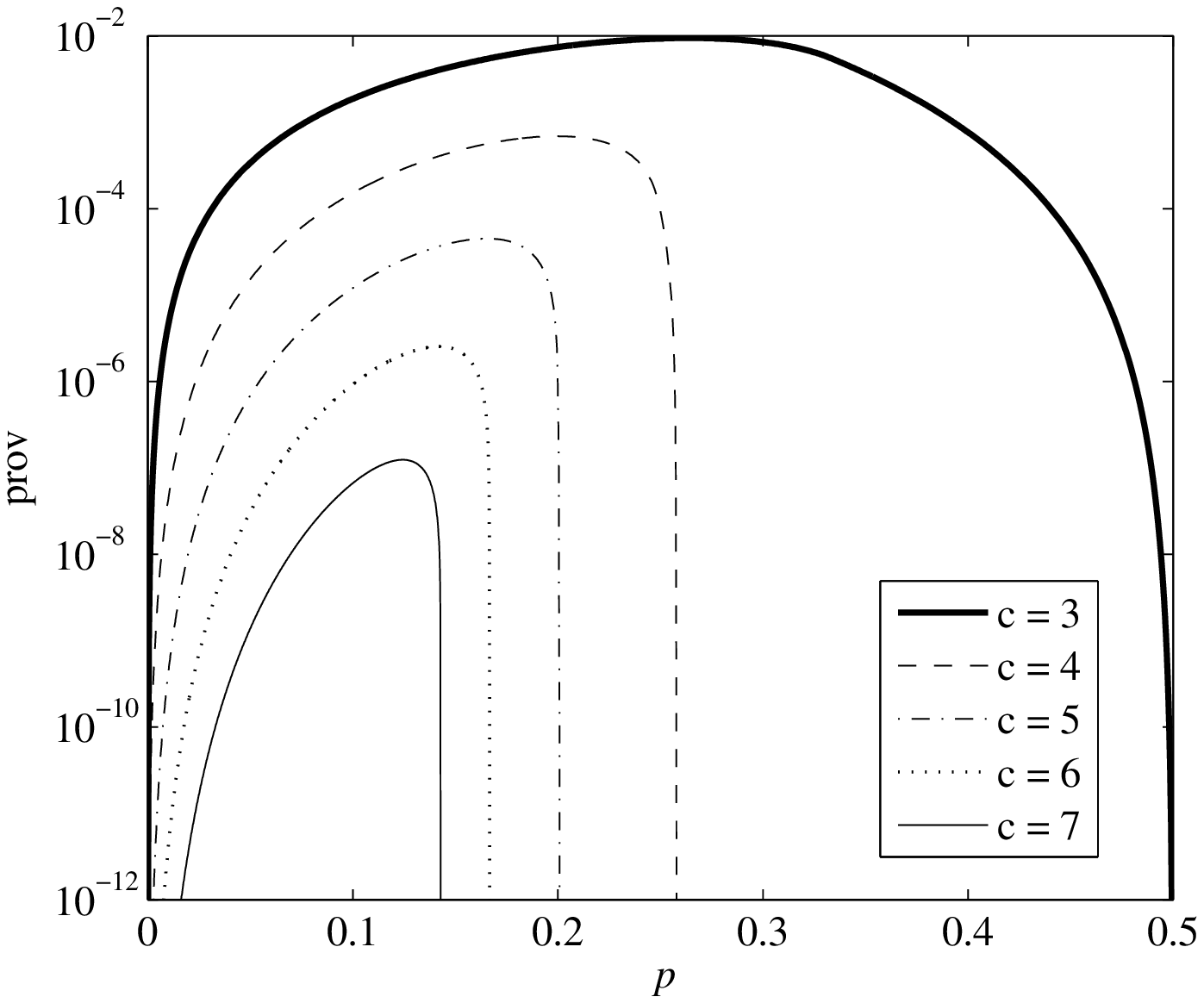, width=1.0\textwidth}}
        \centerline{(b)}
    \end{minipage}
    \caption{Simple decoder against worst Class-D collusion:
    plot of $\boldtheta^*$ according to \eqref{eq:linesearch} (a), and plot of $r_{\simple}^D(c,p)$ (b).}
    \label{fig:opt-simple-knownp}
\end{figure*}

Let us makes some comments about the optimal Class-D attack.

This attack may seem somewhat counterintuitive. The
simplest strategy when the colluders know $p$ would be to generate
a new sequence independent from the embedded sequences:
$\theta_\sigma(p)=p$. However, when all the colluders' symbols are
the same, they cannot generate the desired output. This is why
this simple strategy is indeed not the worst.

Fig.~\ref{fig:opt-simple-knownp}(a) shows the value of the optimal parameter $\theta_1(p)$ for $p\notin[\eta_c,1-\eta_c]$. A corollary of Prop.~\ref{prop:optimal-classD-simple} is
that $r^D_{\simple}(c,p) = r^D_{\simple}(c,1-p)$.
Fig.~\ref{fig:opt-simple-knownp}(b) shows $r_{\simple}^D(c,p)$ for $p \in [0, 0.5]$ and a different number of colluders, and Fig.~\ref{fig:rate-simple} shows the achievable rate for the Tardos and flat pdf compared to the rates achievable under the other classes of attacks. Surprisingly, $R_{\simple}^D(c)$ is not null, although its decrease seems to be exponentially fast as for the joint decoder.
For every $c$, the capacity-achieving pdf is a symmetric two Dirac's deltas distribution in the values of $p$ maximizing $r_{\simple}^D(c,p)$. 

In the interval $p\in[\eta_c,1-\eta_c]$, the optimal
collusion strategy is given by any vector $\boldtheta$ in the
intersection between a hyperplane and the feasible set
$\mathcal{P}^D(c)$. Hence, the solution is not unique. Yet the
problem is convex, all the solutions cancel the achievable rate.
Notice that the minimum collusion size for nullifying the achievable rate is $c=3$. As proved in Appendix~\ref{app:null-rate-simple-knownp},
for $c=3$ this can be achieved only for the
singleton $p=1/2$, and the resulting worst collusion is the minority collusion strategy.

These results show the need for time-sharing if we want to be
protected against malicious attack based on Class-D collusion
strategies. For instance, a codebook with a fixed value
$p=1/2$ is a bad idea since colluders can always nullify the rate as long as
they are at least 3.

\section{Conclusion}
In this paper we have carried out a performance assessment of probabilistic traitor tracing codes from an information-theoretic point of view. From our investigation, considering four different classes of attackers with increasing power and two different classes of decoders, several important conclusions can be drawn.

Let us first list the bad news. Not knowing the embedded symbols (e.g. the Class-B, a.k.a. symmetric channel~\cite{Amiri2009:High}, or multimedia scenario~\cite{Furon08IH}) does not make the colluders less powerful (see Prop.~\ref{prop:ClassBJoint1} for the joint decoder, and Conj.~\ref{conj:ClassBsimple1} for the simple decoder). The case of the joint decoder is even more hopeless: the simplest collusion attack, Class-A, is asymptotically optimal.
A mixed result is the following: disclosing the secret time-sharing sequence opens the door to a powerful collusion attack but, surprisingly, it does not render the code completely useless since the achievable rate is indeed strictly positive.

The goods news are seldom: the time-sharing sequence plays a key role in the performance of the probabilistic traitor tracing code, offering a polynomial decrease of the achievable rate instead of an exponential decay. The achievable rate of the simple decoder is not so smaller than the one of the joint decoder. This is good news because the complexity of the simple decoder is in $O(n)$ whereas the one of the joint decoder is almost in $O(n^c)$, and in some scenarios, $n$ can be very large. On the other hand, we have focused in the study of two particular codes, but as we have seen, their performance is really close to that of the optimal code (asymptotically the same), obtained in \cite{Amiri2009:High} for the joint decoder. Furthermore, the codes studied here make use of a fixed time-sharing distribution, whereas for the capacity-achieving codes it is strongly dependent on the number of colluders.

The problem of finding the optimal time-sharing distribution for the simple decoder still remains open. However, the results of this paper suggest that no big improvement will be brought about over the existing Tardos pdf, especially if a large number of colluders is concerned. Our future works will investigate the trade-off between the complexity and the efficiency of the decoder proposing new traitor tracing decoding algorithms.

\section*{Acknowledgements}
The authors are grateful to Dr. Arnaud Guyader for his valuable advise while preparing the revised version of this paper.

\appendices

\section{Proofs of the propositions about the joint decoder}
\label{App:ClassBJoint}

\subsection{Proof of Prop.~\ref{prop:JointClassA}}
\label{sub:ProofClassAJoint}

The expression \eqref{eq:RateJointClassA} of the rate $R^{A}_{\joint}(c)$ can be rewritten in terms of a double expectation:
$$
R^{A}_{\joint}(c) =c^{-1}\expect{P}{\expect{S_{c}}{h_{b}(P)-h_{b}(S_{c})|P=p}}, 
$$ 
where $S_{c}$ is a random variable distributed as a binomial $B(c,p)$ but divided by $c$. Thus, its expectation equals $p$ and its variance $p(1-p)c^{-1}$.
For a given $p\in(0,1)$, we have:
$$
h_{b}(S_{c})=h_{b}(p)+(S_{c}-p)h_{b}^{\prime}(p)+(S_{c}-p)^2h_{b}^{\prime\prime}(p)/2+o_{S_{c}}((S_{c}-p)^2),
$$
where $o_{S_{c}}(\phi(S_{c}))$ means that, statistically, the term $\phi(S_{c})$ is getting smaller and smaller in the sense that $\forall \epsilon>0, \Prob{S_{c}}{|\phi(S_{c})|>\epsilon}\stackrel{c\rightarrow+\infty}{\rightarrow}0$. Taking the expectation conditioned on $P=p$, and the natural logarithm in $h_b(\cdot)$, we have:
\begin{eqnarray}
\expect{S_{c}}{h_b(P)-h_{b}(S_{c})|P=p}&=&-\expect{S_{c}}{S_{c}-p}h_{b}^{\prime}(p)-\expect{S_{c}}{(S_{c}-p)^2}h_{b}^{\prime\prime}(p)/2-o(\expect{S_{c}}{(S_{c}-p)^2}),\nonumber\\
&=&-\frac{p(1-p)}{2c}h_{b}^{\prime\prime}(p)+o(p(1-p)c^{-1})\\
&=&\frac{1}{2\ln(2)c}+p(1-p)o(c^{-1}).
\end{eqnarray}
Therefore, we can write (in natural units) that $R^{A}_{\joint}(c)=\displaystyle{\frac{1}{2\ln(2)c^2}}+\expect{P}{p(1-p)}o(c^{-2})$, and Prop.~\ref{prop:JointClassA} follows.

\subsection{Proof of Prop.~\ref{prop:ClassBJoint1}}
\label{sub:ProofPropClassBJoint1}
The proof uses the following two lemmas:
\begin{lemma}
Class-B collusion attacks have the following property:
\begin{equation}
\Prob{Y}{1|P=p}=1-\Prob{Y}{1|P=1-p}.
\label{eq:AntiSymProbY}
\end{equation}
\end{lemma}
This is easily proven with the change of variables $p'=1-p$ and $\sigma'=c-\sigma$ in~(\ref{eq:ProbY}).

\begin{lemma}
If $f(p)$ is symmetric, i.e. $f(p)=f(1-p),\,\forall p\in(0,1)$, and $q^{(k-1)}(p)=1-q^{(k-1)}(1-p),\,\forall p\in[0,1]$, then $B^{(k)}(\sigma)=1/B^{(k)}(c-\sigma)$.
\end{lemma}
Again, the change of variables $p'=1-p$ and $\sigma'=c-\sigma$ shows that:
\begin{eqnarray}
\expect{P}{\Prob{\Sigma}{\sigma|P=p}}&=&\expect{P}{\Prob{\Sigma}{c-\sigma|P=p}}\\
\expect{P}{\Prob{\Sigma}{\sigma|P=p}\log\frac{1-q^{(k-1)}(p)}{q^{(k-1)}(p)}}&=&-\expect{P}{\Prob{\Sigma}{c-\sigma|P=p}\log\frac{1-q^{(k-1)}(p)}{q^{(k-1)}(p)}}
\label{eq:LemmaRJoint}
\end{eqnarray}
Thus, $B^{(k)}(\sigma)=1/B^{(k)}(c-\sigma)$.

These two lemmas show that the Class B is closed for the iteration defined in the proposed algorithm, for any symmetric pdf $f(p)$. In other words, if $\boldtheta^{(k)}\in\class{B}{c}$, then so is $\boldtheta^{(k+1)}$. Since the Blahut-Arimoto algorithm converges to the minimum achievable rate whatever the initial vector $\boldtheta^{(0)}$, and in particular, for $\boldtheta^{(0)}\in\class{B}{c}$, we can conclude that Class-B colluders can lead the worst case collusion.

\subsection{Proof of Prop.~\ref{th:worstattack-joint-knownp}}
\label{app:worstattack-joint-knownp}

We compute the gradient of the mutual information with respect to
the parameters of the collusion model $\theta_\sigma,\, \sigma \in
[c-1]$. For the first term in the rhs of
(\ref{eq:rate-joint-decoder}):
\begin{equation}
\frac{\partial}{\partial\theta_{\sigma}}H(Y|P=p)=
\Prob{\Sigma}{\sigma|P=p}\log\left(\frac{1-\Prob{Y}{1|P=p}}{\Prob{Y}{1|P=p}}\right).
\label{eq:gradient-HY-joint}
\end{equation}
For the conditional entropy:
\begin{eqnarray}
    \frac{\partial}{\partial\theta_{\sigma}} H(Y|\Sigma,P=p) & = & \sum_{\sigma=0}^c   \frac{\partial}{\partial\theta_{\sigma}}\left(H(Y|\Sigma=\sigma)\Prob{\Sigma}{\sigma|P=p} \right) \nonumber\\
    & = &  \Prob{\Sigma}{\sigma|P=p}\log\left(\frac{1-\theta_\sigma}{\theta_\sigma}\right)
    \label{eq:gradient-condentropy-Y-joint}
\end{eqnarray}
By combining (\ref{eq:gradient-HY-joint}) and
(\ref{eq:gradient-condentropy-Y-joint}), we obtain the expression
\begin{eqnarray}
    \frac{\partial}{\partial\theta_{\sigma}} I(Y;\Sigma|P=p) =    \Prob{\Sigma}{\sigma|P=p}\log\left(\frac{(1-\Prob{Y}{1|P=p})\theta_\sigma}{(1-\theta_\sigma)\Prob{Y}{1|P=p}}\right). \nonumber
\end{eqnarray}
Hence, in order to cancel the gradient we need to fulfill
$\Prob{Y}{1|P=p} = \theta_\sigma =
\theta^*,\,\forall\,\sigma\in[c-1]$. This condition can be written
as
\begin{eqnarray}
    \theta^* & = & \Prob{Y}{1|P=p,\boldcaptheta=\boldtheta} =\theta^*\sum_{\sigma=1}^{c-1} \Prob{\Sigma}{\sigma|P=p} + \Prob{\Sigma}{c|P=p}  \nonumber\\
    & = & \theta^*(1-(1-p)^c-p^c) + p^c.
\end{eqnarray}
Working out this last expression, the Class-D worst case collusion
results in the one stated in Prop.~\ref{th:worstattack-joint-knownp}.

\subsection{Proof of Prop.~\ref{th:joint_capacity_classD}}
\label{app:joint_capacity_classD}

According to Section~\ref{sec:class-z-attack}, the rate can be written as
$R^D_{joint}(c) = \expect{P}{r^D(c,P)}$. Our objective here is to show
that $R^D_{joint}(c)$ is maximized for $f(p) = \delta(p-1/2)$. We first
insert \eqref{eq:optimal-attack-joint-knownp-final} in
\eqref{eq:rjoint-kl} to obtain, after simplifications,
\begin{eqnarray}
	r^D(c,p) = \frac{1}{c}\left( p^c \log\left( \frac{(1-p)^c}{p^c} +
1\right) + (1-p)^c \log\left( \frac{p^c}{(1-p)^c} + 1 \right) \right),\;
\mbox{for } p \in [0,1].
	\nonumber
\end{eqnarray}
This function is not negative and symmetric: $r^D(c,p)=r^D(c,1-p)$. Its
derivative in $p$ can readily be shown to be given, after pertinent
simplications, by the following expression:
\begin{equation}
r^{D\prime}(c,p)=(p^{c-1}+(1-p)^{c-1})\left(\frac{1-\theta^\star(p)}{1-p}\log(1-\theta^\star(p))-\frac{\theta^\star(p)}{p}\log(\theta^\star(p))\right)
\end{equation}
This function clearly cancels in $p\in\{0,1/2,1\}$. We only focus on the
interval $p\in (0,1/2)$ to show that it never cancels again.
Then $(1-p)^{-1}<2$ and $-p^{-1}<-2$. Since $\log(1-\theta^\star(p))$ and
$\log(\theta^\star(p))$ have negative values, we have:
\begin{equation}
r^{D\prime}(c,p)>
2(p^{c-1}+(1-p)^{c-1})\left((1-\theta^\star(p))\log(1-\theta^\star(p))-\theta^\star(p)\log(\theta^\star(p)\right)
\end{equation}
Knowing that $0<\theta^{\star}(p)<1/2$ on the interval $p\in (0,1/2)$ and that
$(1-x)\log(1-x)-x\log(x)$ is positive for $0<x<1/2$, it appears that the
derivative is strictly positive over $p\in (0,1/2)$.
This proves that $r^D(c,p)$ is strictly increasing on this interval and
reaches a unique maximum in $p=1/2$.

\section{Proofs of the propositions about the simple decoder}
\label{App:ClassBsimple}

\subsection{Proof of Lemma~\ref{lemma:null-rate-simpledec}}
\label{app:null-rate-simple-knownp}

We first redefine (\ref{eq:CondProb1}) and (\ref{eq:CondProb0}) as:
\begin{eqnarray}
    \Prob{Y}{1|X=1,P=p} & = & \boldtheta^T \bm q_{\Sigma1}, \nonumber\\
    \Prob{Y}{1|X=0,P=p} & = & \boldtheta^T \bm q_{\Sigma0},
    \label{eq:probsy}
\end{eqnarray}
with $\bm q_{\Sigma1}$ and $\bm q_{\Sigma0}$ defined in \eqref{eq:definition_qsigma1} and \eqref{eq:definition_qsigma0}. The $\sigma$-th component of $\bm q_{\Sigma1}$ and $\bm
q_{\Sigma0}$ is also given by $\Prob{\Sigma}{\sigma|P=p}\sigma/(cp)$
and $\Prob{\Sigma}{\sigma|P=p}(c-\sigma)/(c(1-p))$, respectively.
A necessary and sufficient condition for achieving $I(Y;X|P=p)=0$
is that $\Prob{Y}{y|X=1,P=p}=\Prob{Y}{y|X=0,P=p}.$ Taking into
account the identities above, this can be expressed as
\begin{eqnarray}
    J(\boldtheta) \triangleq \boldtheta^T (\bm q_{\Sigma1}- \bm q_{\Sigma0}) = 0.
    \label{eq:condition-null-mutinf}
\end{eqnarray}
Hence, we must find at least one vector $\boldtheta \in\mathcal{P}^D(c)$ orthogonal to
$(\bm q_{\Sigma1}- \bm q_{\Sigma0})$, with $\mathcal{P}^D(c)$ defined in
\eqref{eq:classDcollusion}.
Taking into account the linearity of the scalar product, and that $\theta_0=0$, $\theta_c=1$ by the marking assumption, $J(\boldtheta,p)$ can be written as a convex conical combination of scalar products: 
\begin{eqnarray}
	J(\boldtheta,p) = \rho_c(p) + \sum_{i=1,\ldots,c-1} \theta_i \cdot \rho_i(p),\; \theta_i \in [0,1]
	\label{eq:conical-comb}
\end{eqnarray}
where 
\begin{eqnarray}
	\rho_i(p) = \bm e_{i+1}^T(\bm q_{\Sigma1}- \bm
q_{\Sigma0}) = \nchoosek{c}{i}p^{i-1}(1-p)^{c-i-1}(i/c-p), \forall i\in[c],
\end{eqnarray}
with $\bm e_i$ the $i$th canonical vector.

Note that, on the interval $[0,1/c]$, only $\rho_{0}(p)$ has negative values, but this term is excluded from the sum since $\theta_{0}=0$. Hence, \eqref{eq:condition-null-mutinf} can't be satisfied on this interval. In the same way, $\rho_{1}(p)$ is the only term producing negative values over the interval $[1/c,2/c]$.
Therefore, we have the lower bound:
\begin{eqnarray}
J(\boldtheta,p) \geq J(\boldtheta_{low},p) = \rho_{1}(p)+\rho_{c}(p),\; p\in[1/c,2/c],
\label{eq:lowbound-proof}
\end{eqnarray}
with $\boldtheta_{low}=(0,1,0,\ldots,0,1)^T$.

For $c=3$, $J(\boldtheta_{low},p)=(2p-1)^2\geq 0$. Therefore, it is not possible to find any vector $\boldtheta \in\mathcal{P}^D(c)$ orthogonal to $(\bm q_{\Sigma1}- \bm q_{\Sigma0})$, except if $p=1/2$ and then $\boldtheta_{low}=(0,1,0,1)^T$ (i.e. a minority vote) cancels the mutual information. 

For $c>3$, $J(\boldtheta_{low},p) = (1-p)^{c-2}(1-cp)+p^{c-1}$ is positive for $p=1/c$ and negative for $p=2/c$. Therefore, there exists some $\eta_{c}\in[1/c,2/c]$ such that, for $p>\eta_c$, $J(\boldtheta_{low},p)$ is negative. 
The vector $\boldtheta=(0,\ldots,0,1)$ gives $J(\boldtheta,p) = \rho_{c}(p) > 0\,\forall\,p\in[0,1]$. Therefore, by continuity, there exists at least one vector $\boldtheta$ satisfying ~\eqref{eq:condition-null-mutinf} and thus canceling the mutual information. Conversely, for $p<\eta_{c}$, \eqref{eq:condition-null-mutinf} cannot be satisfied. Moreover, $J(\boldtheta_{low},p)$ can be shown to be negative in the whole interval $[\eta_c,1/2]$, for which $\rho_c(p)$ is strictly positive. Hence, \eqref{eq:condition-null-mutinf} can be satisfied in this whole interval.

As $c$ increases, $\eta_c$ asymptotically approaches $1/c$ (see Tab.~\ref{tab:etac}). Intuitively, this is explained as follows: the behavior of $J(\boldtheta_{low},p)$ over $[1/c,2/c]$ is dominated by the term $\rho_{1}(p)$ which is strictly decreasing on this interval and equaling zero in $p=1/c$. This justifies why $\lim_{c\rightarrow\infty}\eta_{c}-1/c=0$. To be more rigorous, let us first denote $u=cp$ with $u\in[1,2]$. In the interval $p\in[1/c,2/c]$, the polynomial $J(\boldtheta_{low},p)$ in \eqref{eq:lowbound-proof} can be expressed as $J(\boldtheta_{low},u) = (1-u/c)^{c-2}(1-u)+(u/c)^{c-1}$. 
For $u\in [1,2]$, $J(\boldtheta_{low},u) \leq (1-2/c)^{c-2}(1-u)+(2/c)^{c-1}$ which cancels for $u_{c}=1 + ((2/c)^{c-1})/(1-2/c)^{c-2}$, and $\lim_{c\rightarrow\infty} u_{c}=1$. Since $1/c\leq\eta_{c}\leq u_{c}/c$, then $\lim_{c\rightarrow\infty} \eta_{c}-1/c=0$. From the expresion of $u_{c}$, we can write $\eta_{c}=1/c + O(1/c^{c})=1/c+o(1/c)$ since $c>2$.



\begin{table}
\begin{center}
\caption{Values of $\eta_{c}$.} 
\label{tab:etac}
  \begin{tabular}{|l||c|c|c|c|c|c|c|}
    \hline
    $c$ & 3 & 4 & 5 & 6 & 10 & 15 & 20 \\
    \hline
	$\eta_{c}-1/c$ & $1.7*10^{-1}$ & $7.8*10^{-3}$ & $6.3*10^{-4}$ & $4.5*10^{-5}$ & $2.3*10^{-10}$ & $<\epsilon$ & $<\epsilon$ \\
    \hline 
  \end{tabular}
  \end{center}
\end{table}

The same rationale holds on the interval $[1-2/c,1-1/c]$, where all the scalar products have negative values except $\rho_{c-1}(p)$, hence a lower bound for \eqref{eq:conical-comb} is:
$$
J(\boldtheta, p) \geq \sum_{i=1}^{c-2}\rho_{i}(p)+\rho_{c}(p).
$$
We can simplify the lower bound into: $p^{c-2}(1-c(1-p))+(1-p)^{c-1}$, which is the symmetric version of the first bound. Hence, for $p>1-\eta_{c}$, it is not possible to cancel the mutual information.

\subsection{Proof of Lemma~\ref{lemma:opt-strat-simpledec-outrange}}
\label{app:LemmaClassDsimple}

For the sake of simplicity, we replace the notation $P=p$ by $p$ and $X=x$ by $x$ in the sequel.
This appendix concerns the worst case for values of $p$ outside the interval $[\eta_c,1-\eta_c]$, i.e. $\Prob{Y}{1|p}\neq\Prob{Y}{1|0,p}\neq\Prob{Y}{1|1,p}$ necessarily.
Denote by $\nabla I(Y;X|p)(\sigma)$ the derivative with respect to the parameter of the collusion model $\theta_\sigma$:
\begin{eqnarray}
\nabla I(Y;X|p)(\sigma)&=&\Prob{\Sigma}{\sigma|p}h_b^\prime(\Prob{Y}{1|p})-p\Prob{\Sigma}{\sigma|1,p}h_b^\prime(\Prob{Y}{1|1,p})\nonumber\\
&-&(1-p)\Prob{\Sigma}{\sigma|0,p}h_b^\prime(\Prob{Y}{1|0,p})
\label{eq:DerivMutInfoKnowing}
\end{eqnarray}
with $h_b^\prime(x)=\log\frac{1-x}{x}$, the derivative of the binary entropy which is strictly decreasing. This simplifies in
\begin{eqnarray}
\nabla I(Y;X|p)(\sigma)&=&\Prob{\Sigma}{\sigma|p}\left(h_b^\prime(\Prob{Y}{1|p})-\frac{\sigma}{c}h_b^\prime(\Prob{Y}{1|1,p})-\frac{c-\sigma}{c}h_b^\prime(\Prob{Y}{1|0,p})\right)\label{eq:DerivMutInfoKnowing2}\\
&=&\Prob{\Sigma}{\sigma|p}K_1(p,c)(\sigma-K_2(p,c))
\label{eq:K2}
\end{eqnarray}
with
\begin{eqnarray}
K_1(p,c)&=&c^{-1}(h_b^\prime(\Prob{Y}{1|0,p})-h_b^\prime(\Prob{Y}{1|1,p}))\\
K_2(p,c)&=&c\frac{h_b^\prime(\Prob{Y}{1|0,p})-h_b^\prime(\Prob{Y}{1|p})}{h_b^\prime(\Prob{Y}{1|0,p})-h_b^\prime(\Prob{Y}{1|1,p})}
\label{eq:K_2def}
\end{eqnarray}

For the parameters of the collusion attack $\boldtheta^*$ (except $\theta_0=0$ and $\theta_c=1$), there are three possibilities :
\begin{itemize}
    \item if $\theta^*_\sigma\in]0,1[$ then $\nabla I(Y;X|p)(\sigma)=0$,
    \item if $\theta^*_\sigma=0$, then $\nabla I(Y;X|p)(\sigma)\geq0$,
    \item if $\theta^*_\sigma=1$, then $\nabla I(Y;X|p)(\sigma)\leq0$.
\end{itemize}

From now on, we detail the case of $p\in[0,\eta_{c})$, but the case of the interval $[1-\eta_{c},1]$ can be deduced by symmetry. Appendix~\ref{app:null-rate-simple-knownp} shows that $J(\boldtheta,p)$, which was defined in \eqref{eq:condition-null-mutinf}, is positive for $p\in[0,\eta_{c})$. This implies that $\Prob{Y}{1|0,p}<\Prob{Y}{1|p}<\Prob{Y}{1|1,p}$. Since $h_b^\prime(x)$ is strictly decreasing, $0<K_{1}(p,c)$ and $0\leq K_{2}(p,c)\leq c$. Therefore, if $\theta_{\sigma}^{\star}=1$ (resp. 0) then $\sigma\leq K_{2}(p,c)$ (resp. $K_{2}(p,c) \leq \sigma$), and if $\theta_{K_{2}(p,c)}^{\star}\in[0,1]$ then $K_{2}(p,c)\in\{1,2,...,c-1\}$. In the sequel we look for closer bounds on $K_{2}(p,c)$.

Bound \#1: $1\leq K_{2}(p,c)\leq c$. This amounts to prove that $(0,\ldots,0,1)$ and $(0,1,\ldots,1)$ do not minimize $I(Y;X|p)$. The first choice raises a contradiction: $\Prob{Y}{1|0,p}=0$ implies that $\nabla I(Y;X|p)(\sigma)<0$ and necessarily $\theta^{\star}_{\sigma}=1$. The second choice also leads to a contradiction: $\Prob{Y}{1|1,p}=1$ implies that $\nabla I(Y;X|p)(\sigma)>0$ and necessarily $\theta^\star_{\sigma}=0$.

Bound \#2: $1\leq K_{2}(p,c)<2$.  Let us define $A(p)\triangleq\boldtheta^T\nabla I(Y;X|p)$. According to~\eqref{eq:DerivMutInfoKnowing}, it follows that:
\begin{equation}
A(p)=g(\Prob{Y}{1|p})-pg(\Prob{Y}{1|1,p})-(1-p)g(\Prob{Y}{1|0,p}),
\label{eq:AKnowing}
\end{equation}
with $g(x)=xh_b^\prime(x)$. As $g(x)$ is strictly concave, $A(p)>0$ for any $p\in (0,\eta_c)$.
With the help of \eqref{eq:K2}, $A(p)$ can also be written as:
\begin{equation}
A(p)=K_{1}(p,c)\left(c\cdot p^{c}+\sum_{0<\sigma\leq K_{2}(p,c)}\sigma \Prob{\Sigma}{\sigma|p}-K_{2}(p,c)\left(p^c+\sum_{0<\sigma\leq K_{2}(p,c)}\Prob{\Sigma}{\sigma|p}\right)\right).
\end{equation}
Since $A(p)>0$, then $K_{2}(p,c)<B(K_{2}(p,c),p,c)$ with
\begin{equation}
B(K,p,c)=\frac{c\cdot p^{c}+\sum_{0<\sigma\leq K}\sigma \Prob{\Sigma}{\sigma|p}}{p^c+\sum_{0<\sigma\leq K}\Prob{\Sigma}{\sigma|p}}
\end{equation}
In the following we will make use of the next lemma, which is proved at the end of the appendix.
\begin{lemma}
For $1\leq K$, $B(K,p,c)\leq B(K+1,p,c)$.
\label{lemma:proof-lemma2}
\end{lemma}

Therefore, $K_{2}(p,c)<B(c-1,p,c)=cp/(1-(1-p)^c)$. This last function is increasing with $p$: $K_{2}(p,c)<B(c-1,3/2c,c)=3/2(1-(1-3/2c)^c),\,\forall p\in(0,\eta_{c}]$ since $\eta_{c}$ is never bigger than $3/2c$. This function is increasing with $c$ and converges to $3/2(1-e^{-3/2})\approx1.93$. Thus, combining this result with the Bound \#1, we have $1\leq K_{2}(p,c)<2$ for $p\leq\eta_{c}$, and $\boldtheta^\star$ has the form $(0,\theta_{1}(p),0,\ldots,0,1)^{T}$ on the interval $(0,\eta_{c}]$, as expressed in the statement of Lemma~\ref{lemma:opt-strat-simpledec-outrange}. The remaining of the proof deals with a refinement of this result in the interval $p\in[c^{-1},\eta_{c}(c)]$.

Bound \#3: $K_{2}(p,c)>cp$ when $p\in[c^{-1},\eta_{c}(c)]$.
At most, when $\theta_{1}=1$, $\Prob{Y}{1|1,p}=(1-p)^{c-1}+p^{c-1}$ which is a decreasing function over $[0,1/2]$. Taken in $p=1/c$, we have decreasing values with $c$ which are all lower than $1/2$ when $c\geq4$. Therefore, $\forall c\geq4,\,\forall p\in[c^{-1},\eta_{c}]$, $\Prob{Y}{1|1,p}$, but also $\Prob{Y}{1|p}$ and $\Prob{Y}{1|0,p}$, lies in the interval $[0,1/2]$ where the function $h_b^\prime(x)$ is strictly convex:
\begin{equation}
h_b^\prime(\Prob{Y}{1|p})<ph_b^\prime(\Prob{Y}{1|1,p})+(1-p)h_b^\prime(\Prob{Y}{1|0,p}).
\label{eq:cond_k_2}
\end{equation}
Using \eqref{eq:K_2def} and \eqref{eq:cond_k_2}, it results that $K_{2}(p,c)>cp$. Hence, we can conclude that $\boldtheta^\star=(0,1,0,\ldots,0,1)^T$ when $p\in[c^{-1},\eta_{c}]$ and $c\geq4$.

We now address the case of $c=3$ and we verify that $\boldtheta^\star=(0,1,0,1)^T$ when $p\in[1/3,1/2]$ (recall from Appendix~\ref{app:null-rate-simple-knownp} that $\eta_3 = 1/2$). With this choice of $\boldtheta^{\star}$, $\Prob{Y}{1|1,p}=1-\Prob{Y}{1|0,p}$, which yields that $K_{2}(p,3)\geq1$ if $h_b^\prime(\Prob{Y}{1|p})/h_b^\prime(\Prob{Y}{1|0,p})\leq 1/3$. Since both derivatives equal $0$ in $p=\eta_{3}=1/2$, we apply l'Hôpital's rule twice to obtain:
$$
\lim_{p\rightarrow1/2}\frac{h_b^\prime(\Prob{Y}{1|p})}{h_b^\prime(\Prob{Y}{1|0,p})}=\lim_{p\rightarrow1/2}\frac{d\Prob{Y}{1|p}/d p}{d\Prob{Y}{1|0,p}/dp}=\lim_{p\rightarrow1/2}\frac{d^2\Prob{Y}{1|p}/d^2 p}{d^2\Prob{Y}{1|0,p}/d^2 p}=0.
$$ 
This shows that $K_{2}(p,3)\geq 1$ in an interval $[\alpha,1/2]$. Remarkably, it appears that $\alpha=1/3$.\\



{\it Proof of Lemma~\ref{lemma:proof-lemma2}}: For fixed $(p,c)$, if $K\geq c$, $B(K,p,c)$ is constant. Otherwise, $B(K+1,p,c)-B(K,p,c)$ has the same sign as $\Delta=K+1-(cp^c+\sum_{\sigma\leq K}\sigma\Prob{\Sigma}{\sigma|p})/(p^c+\sum_{\sigma\leq K}\Prob{\Sigma}{\sigma|p})$. The successive derivations hold:
\begin{eqnarray}
\Delta&=&K+1-c\frac{p^c}{p^c+\sum_{\sigma\leq K}\Prob{\Sigma}{\sigma|p}}-\frac{\sum_{\sigma\leq K}\sigma\Prob{\Sigma}{\sigma|p}}{p^c+\sum_{\sigma\leq K}\Prob{\Sigma}{\sigma|p}}\\
&>& 1-c\frac{p^c}{p^c+\sum_{\sigma\leq K}\Prob{\Sigma}{\sigma|p}}>1-c\frac{p^c}{p^c+cp(1-p)^{c-1}}>0
\end{eqnarray}
The last inequality holds provided that $p<g(c)=1/(1+(1-1/c)^{1/(c-1)}$.
$g(.)$ is a decreasing function and $\lim_{c\rightarrow\infty}g(c)=1/2$. Thus, $\forall c\geq3$, we have $p\leq\eta_{c}\leq g(c)$, which proves the lemma.
  
\bibliographystyle{IEEEtran}
\bibliography{fingerprinting_biblio}

\end{document}